\documentclass[envcountsame,envcountsect,draft]{llncs}

\usepackage{mathpartir}
\usepackage{amsmath,amssymb,stmaryrd,amsbsy}
\usepackage{upgreek}
\usepackage{graphicx}
\usepackage{rotating}
\usepackage{xspace}
\usepackage[all]{xy}
\usepackage{array}
\usepackage{fancyhdr}
\usepackage{url}

\usepackage{float}
\floatstyle{boxed} 
\restylefloat{figure}

\usepackage{etapsFs}
\usepackage{jbw-thm}

\date{2011-10-14}

\title{Expansion for Universal Quantifiers}

\author{Sergue\"i Lenglet \inst{1} \thanks{The author is supported by the
    Alain Bensoussan Fellowship Programme} \and J. B. Wells \inst{2}}

\institute{University of Wroc\l{}aw \and Heriot-Watt University}

\begin{document}

\pagestyle{fancy}  	 
\lhead{} 	
\chead{} 	
\rhead{} 	
\lfoot{} 
\cfoot{\thepage} 	
\rfoot{} 	
\renewcommand{\headrulewidth}{0pt}

\maketitle
\thispagestyle{fancy}  	 

\begin{abstract}
  \emph{Expansion} is an operation on typings (i.e., pairs of typing
  environments and result types) defined originally in type systems for the
  $\uplambda$-calculus with intersection types in order to obtain principal (i.e.,
  most informative, strongest) typings.  In a type inference scenario, expansion
  allows postponing choices for whether and how to use non-syntax-driven typing
  rules (e.g., intersection introduction) until enough information has been
  gathered to make the right decision. Furthermore, these choices can be
  equivalent to inserting uses of such typing rules at deeply nested positions
  in a typing derivation, without needing to actually inspect or modify (or even
  have) the typing derivation. Expansion has in recent years become simpler due
  to the use of \emph{expansion variables} (e.g., in System E).

  This paper extends expansion and expansion variables to systems with
  $\forall$-quantifiers.  We present \systemfs, an extension of System F with
  expansion, and prove its main properties. This system turns type inference
  into a constraint solving problem; this could be helpful to design a modular
  type inference algorithm for System F types in the future.
\end{abstract}

\section{Introduction}

\subsection{Background and Motivation}

\paragraph*{Polymorphism and principal typings.} Many practical uses of type
systems require \emph{polymorphism}, i.e., the possibility to reuse a generic
piece of code with different types. Type systems most commonly provide
polymorphism through $\forall$-\emph{quantifiers}, like in the Hindley-Milner
(HM) type system \cite{Milner78} and in System F \cite{Reynolds74,Girard72}, but
can also use other methods like \emph{intersection types}
\cite{CoppoDV80}. Systems with $\forall$-quantifiers assign general type schemes
that can be instantiated to more specific types; for example, the identity
function can be typed with $\forall \tvar.(\fun \tvar \tvar)$, and then used
with types $\fun{\mathsf{int}}{\mathsf{int}}$ or
$\fun{\mathsf{real}}{\mathsf{real}}$ when applied respectively to an integer or
a real. Systems with intersection types list the different usage types of a
term; if the identity function is applied exactly twice in a code fragment, once to an integer and
once to a real, then its type will be $(\fun{\mathsf{int}}{\mathsf{int}}) \inter
(\fun{\mathsf{real}}{\mathsf{real}})$.

Type systems with $\forall$-quantifiers are very popular, but they often lack
\emph{principal typings} \cite{Wells02}, i.e., strongest, most informative
typings (a typing is usually a pair of a type environments and a result type).
Wells \cite{Wells02} proved that HM and System F do not have principal
typings. It is important not to confuse this notion with the (weaker) one of
``principal types'' defined for the HM type system in which typable terms admit
a strongest result type for each fixed type environment. Principal typings are
crucial for \emph{compositional} type inference, where types for terms are found
using only the analysis results of the immediate sub-components, which can be
inspected independently and in any order. Compositional type inference helps in
performing separate analysis of program modules, and therefore helps with
separate compilation. Note that the Damas-Milner algorithm \cite{DamasM82} for
HM is not fully compositional: to give a type for a let-binding
$\mathsf{let}~x=e_1~\mathsf{in}~e_2$, the algorithm must infer first a type for
$e_1$, and then use the result to type $e_2$.

\paragraph*{Expansion and expansion variables.} In contrast, type systems with
intersection types usually have principal typings \cite{CoppoDV80}. In such
systems, admissible typings are obtained from a principal one using
\emph{expansion} (in addition to substitution and weakening). We present this
mechanism through an example, taken from \cite{CarlierW05}. Consider the
following $\lambda$-terms:
$$\term_1 = \lamb \expvar {x ~ (\lamb y {y ~ z})} \qquad \term_2 = \lamb g
{\lamb x {g ~ {(g ~ x)}}}$$
One can compute the following principal typings for these terms in the
type system of Coppo, Dezani, and Veneri \cite{CoppoDV80}. 
\begin{eqnarray*}
  \term_1 &:& \langle \envenum {\eassign z a} \vdash \underbrace{(\fun{(\fun
      {(\fun a b)} b)} c)}_{\typ_1} \rightarrow c \rangle \\
  \term_2 &:& \langle \emptyenv \vdash \underbrace{\fun{((\fun e f)
      \inter (\fun d e))}{(\fun d f)}}_{\typ_2} \rangle 
\end{eqnarray*}
Following \cite{CarlierW05}, we write $\term \mathrel : \typing \envvar \typ$
for the assignment of type $\typ$ under type environment $\envvar$ (often written
$\envvar \vdash \term \mathrel : \typ$ in the literature). To type the application $\term_1 ~
\term_2$, we must somehow ``unify'' $\typ_1$ and $\typ_2$. We cannot do this
by simple type substitutions, replacing type variables by types; we have a clash
between type $\fun {(\fun a b)} b$ and type $(\fun e f) \inter (\fun d e)$. We
cannot unify these types by removing the intersection, using
idempotence $\typ \inter \typ = \typ$; we would have to solve the equation $\fun
a b=b$, which does not have a solution in absence of recursive types.

This inference problem can be solved by introducing an intersection in the
typing of $\term_1$, using \emph{expansion}.
$$
\term_1 \mathrel : \langle \envenum {\eassign z {a_1 \inter a_2}} \vdash 
(\fun{(\fun
    {(\fun {a_1}{b_1})}{b_1} \inter \fun
    {(\fun {a_2}{b_2})}{b_2})} c) \rightarrow c \rangle
$$
We can then unify the two types as required by applying the substitution $e :=
\fun{a_1}{b_1} , \assign f {b_1}, \assign d {\fun {a_2}{\fun {a_1}{b_1}}}, b_2
:= \fun {a_1}{b_1}, \assign c {\fun {(\fun {a_2}{\fun{a_1}{b_1}})}{b_1}}$

The expansion operation simulates on typings the use of an intersection
introduction typing rule at a nested position in the typing derivation. The
above expansion on the typing of $\term_1$ transforms the typing derivation on
the left in the figure below into the derivation on the right (we write $\app$
for the application typing rule, $\uplambda$ and $\inter$ for respectively
abstraction and intersection introductions),

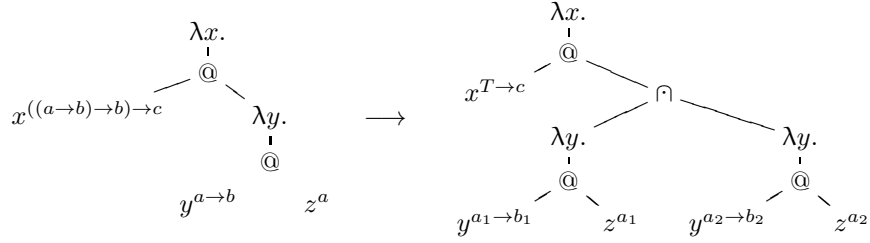
\begin{center}
\newcolumntype{V}{>{\centering\arraybackslash} m{4.5cm}}
\newcolumntype{W}{>{\centering\arraybackslash} m{1cm}}
\newcolumntype{X}{>{\centering\arraybackslash} m{6cm}}
\begin{tabular}{VWX}
  $
  \xymatrix@C=0pt@R=2pt{ &\lamb x {} \ar@{-}[d]& \\
    & \app \ar@{-}[dr] \ar@{-}[dl] & \\
    x^{\fun{(\fun {(\fun a b)} b)} c} & & \lamb y {} \ar@{-}[d] \\
    & & \app \ar@{-}[dr] \ar@{-}[dl] \\
    & y^{\fun a b} & & z^a}
  $
  & $\longrightarrow$ &
  $
  \xymatrix@C=0pt@R=2pt{& &\lamb x {} \ar@{-}[d]& \\
    & & \app \ar@{-}[drr] \ar@{-}[dl] & \\
    & x^{\fun \typ c} & & & \inter \ar@{-}[drr] \ar@{-}[dll] \\ 
    & & \lamb y {} \ar@{-}[d] & & & & \lamb y {} \ar@{-}[d] \\
    & & \app \ar@{-}[dr] \ar@{-}[dl] & & & & \app \ar@{-}[dr]
    \ar@{-}[dl] \\
    & y^{\fun {a_1}{b_1}} & & z^{a_1} & 
    & y^{\fun{a_2}{b_2}} & & z^{a_2}}
  $
\end{tabular}
\end{center}
where $\typ = (\fun {(\fun {a_1}{b_1})} {b_1}) \inter (\fun {(\fun {a_2}{b_2})}
{b_2})$.

Earlier definitions of expansion \cite{CoppoDV80,RonchiV84} are quite difficult
to follow and to implement. \emph{Expansion variables} (or E-variables) were
introduced by Kfoury and Wells in System I \cite{KfouryW04} to simplify
expansion application. The construct has then been improved in System E
\cite{CarlierPWK04}. An E-variable $\evar$ is a placeholder for unknown uses of
typing rules such as $\inter$-introduction. For example, the following typing
derivation for $\term_1$
$$
  
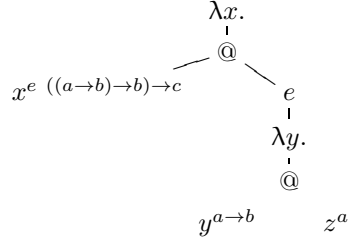
\xymatrix@C=0pt@R=2pt{ &\lamb x {} \ar@{-}[d]& \\
    & \app \ar@{-}[dr] \ar@{-}[dl] & \\
    x^{\fun{\etyp \evar {(\fun {(\fun a b)} b)}} c} & & \evar \ar@{-}[d] \\
    & & \lamb y {} \ar@{-}[d] \\
    & & \app \ar@{-}[dr] \ar@{-}[dl] \\
    & y^{\fun a b} & & z^a}
$$
generates this typing:
$$
  \term_1 \mathrel : \typing {\envenum {\eassign z {\etyp \evar a}}}{\fun{(\fun{\etyp
      \evar {(\fun
      {(\fun a b)} b)}} c)} c}
$$
Note that the variable $\evar$ is introduced in the result type as well as in
the type environment. One can then perform the previous expansion by replacing
$\evar$ by the \emph{expansion term} $(\assign \tvar {\tvar_1} \sepsubs \assign b
{b_1}) \inter (\assign \tvar {\tvar_2} \sepsubs \assign b {b_2})$, which 
introduces an intersection $\inter$ at the $\evar$ position and applies a
different substitution for each branch of the intersection. We then obtain the
desired typing with intersection, given above.

\paragraph*{Motivation.} The idea behind expansion is fairly general, even if it
has been defined only in systems with intersection types. It allows postponing
the uses of non-syntactic typing rules, i.e., rules that are not driven by the
syntax of terms, such as $\inter$-introduction, but also $\forall$-introduction
and $\forall$-elimination. This is helpful in type inference scenarios:
constructor introductions or eliminations can be delayed until all necessary
information has been gathered. In the above example, we introduce an
intersection in the typing of $\term_1$ only when we have to, when applying
$\term_1$ to $\term_2$. We want to bring this possibility of delaying the choice
of uses of typing rules to type system with $\forall$-quantifiers, to see how
(compositional) type inference could benefit from this property. We present an
extension of System F with an expansion mechanism, called \systemfs. Before
going into the details of its syntax in Section \ref{s:syntax}, we first
informally introduce \systemfs and point out the main differences between its
expansion mechanism and the one of System E.

\subsection{Overview of \systemfs}

\paragraph*{Quantifier introduction.} Assume that we have the following
typings for the terms $\term_1$ and $\term_2$ given above.
\begin{eqnarray*}
  \term_1 &:& \langle \envenum {\eassign z a} \vdash \underbrace{(\fun{(\fun
      {(\fun a b)} b)} c)}_{\typ_1} \rightarrow c \rangle \\
  \term_2 &:& \langle \emptyenv \vdash \underbrace{\fun{(\forall e.{(\fun {(\fun
          d e )} e)})}{\fun {(\fun d {\fun d f})} f}}_{\typ_2} \rangle 
\end{eqnarray*}
Suppose we have forgotten $\term_1$ and $\term_2$ (e.g., we have already
compiled them and discarded the source code), and we want to type the
application $\term_1 ~ \term_2$. We need to ``unify'' $\typ_1$ and $\typ_2$. We
cannot unify $\fun {(\fun a b)} b$ and $\forall e.{(\fun {(\fun d e )} e)}$
using only type substitutions, because of the $\forall$-quantifier. This
$\forall$-quantifier is necessary, because the term $g$ is used twice in
$\term_2$ with different usage types. We can solve this problem by introducing
in $\typ_1$ a $\forall$-quantifier over $b$, the scope of which encompasses
$\fun {(\fun a b)} b$. To this end, we introduce an \emph{expansion variable}
$\subv$ at the required position in the typing of $\term_1$ (we use $\subv$
instead of $\evar$ to avoid confusion with the E-variables of System E).
$$
\term_1 \mathrel : \typing {\envenum {\eassign z a}}{\fun{(\fun{\styp \subv
      {\{\tvar \}} {(\fun {(\fun a b)} b)}} c)} c}
$$
Unlike expansion variables in System E, $\subv$ is not introduced in the type
environment; the application of $\subv$ to the typing is asymmetric. We discuss
the role of the superscript $\{ \tvar \}$ below. A $\forall$-quantifier over $b$
can be introduced at the position we want by replacing $\subv$ by the
\emph{expansion term} $\forall b$. This operation corresponds to the following
transformation on derivation trees
\begin{center}
\newcolumntype{V}{>{\centering\arraybackslash} m{5.5cm}}
\newcolumntype{W}{>{\centering\arraybackslash} m{0.75cm}}
\newcolumntype{X}{>{\centering\arraybackslash} m{5.5cm}}
\begin{tabular}{VWX}
$
  \xymatrix@C=0pt@R=2pt{ &\lamb x {} \ar@{-}[d]& \\
    & \app \ar@{-}[dr] \ar@{-}[dl] & \\
    x^{\fun{\styp \subv {\{\tvar\}} {(\fun {(\fun a b)} b)}} c} & & \subv^{\{
      \tvar\} } \ar@{-}[d] \\
    & & \lamb y {} \ar@{-}[d] \\
    & & \app \ar@{-}[dr] \ar@{-}[dl] \\
    & y^{\fun a b} & & z^a}
$
& $\longrightarrow$ &
$
  \xymatrix@C=0pt@R=2pt{& \lamb x {} \ar@{-}[d]& \\
    & \app \ar@{-}[dr] \ar@{-}[dl] & \\
    x^{\fun{(\forall b. (\fun {(\fun a b)} b))} c}
    & & \forall b. \ar@{-}[d] \\
    & & \lamb y {} \ar@{-}[d] \\
    & & \app \ar@{-}[dr] \ar@{-}[dl] \\
    & y^{\fun a b} & & z^a} 
$
\end{tabular}
\end{center}
and generates the typing
$$\term_1 \mathrel : \langle \envenum {\eassign z a} \vdash (\fun{\forall b. (\fun
    {(\fun a b)} b)} c) \rightarrow c \rangle$$
as wished. We can then unify $\fun{\forall b. (\fun {(\fun a b)} b)} c$
with $\typ_2$, by substituting $d$ for $a$ and $\fun{(\fun d {\fun d f})}
f$ for $c$. The key point is we can get the new typing without
needing to build the typing derivation (or have any memory of $\term_1$).

When we introduce a $\forall$-quantifier, we forbid any quantification over type
variables that are free in the type environment. To take this into account, we
keep the set of free variables of the environment as a parameter of the
E-variable. For example, when we introduce $\subv$ in the typing of $\term_1$,
$\tvar$ is the only free variable occurring in the environment; we remember the
set $\{ \tvar \}$ in $\subv^{\{ \tvar \}}$. This prevents any illegal
quantification from happening; replacing $\subv$ by the expansion $\forall
\tvar$ does not introduce a quantification over $\tvar$ in this case and leaves
the typing judgement unchanged.

\paragraph*{Subtyping.} E-variables can be used to perform subtyping as
well. Consider System~F $\forall$-elimination as a subtyping relation: $\forall
\tvar.\typ_1 \leq \subst {\assign \tvar {\typ_2}}{\typ_1}$. Let $\envvar =
\envenum{\eassign {\mathsf{choose}}{\forall \tvar.{(\fun \tvar {\fun \tvar
        \tvar})}} \sepenv \eassign {\mathsf{id}}{\forall \tvar.{(\fun \tvar
      \tvar)}}}$ and suppose we want to type the application $\term =
\mathsf{choose}~\mathsf{id}$ under $\envvar$ (this example is taken from
\cite{BotlanR09}). We can derive the typing $\typing \envvar {\fun{(\forall
    \tvar.{(\fun \tvar \tvar)})}{(\forall \tvar.{(\fun \tvar \tvar)})}}$ for
$\term$; however if we want to apply $\term$ to a term of type $\fun b b$, we
have to redo the type inference on $\term$ to obtain the needed typing $\typing
\envvar {\fun{(\fun b b)}{(\fun b b)}}$.

To avoid this, we add an E-variable $\subv$ on top of the type of $\mathsf{id}$;
we obtain the following typing derivation (nodes marked with a type represent
uses of subtyping, i.e., in our case, instantiations of $\forall$-quantifiers)
$$
\xymatrix@C=0pt@R=2pt{& \app \ar@{-}[dr] \ar@{-}[dl] \\
  \fun \typ {\fun \typ \typ}
  \ar@{-}[d] & & \subv^\emptyset
  \ar@{-}[d] \\
  \mathsf{choose}^{\forall \tvar.{(\fun \tvar {\fun \tvar \tvar})}} & &
  \mathsf{id}^{\forall \tvar.{(\fun \tvar \tvar)}} 
    }
$$
with $\typ = \styp \subv \emptyset {\forall \tvar.(\fun \tvar \tvar)}$, giving typing 
$$\term \mathrel : \typing \envvar {\fun{(\styp \subv \emptyset {\forall
    \tvar.{(\fun \tvar \tvar)}})}{(\styp \subv \emptyset {\forall \tvar.{(\fun \tvar 
      \tvar)}})}}
$$
If we want to apply $\term$ to a term $\term'$ of type $\fun b b$, we utilize 
expansion to introduce the use of subtyping $\forall \tvar.{(\fun \tvar \tvar)}
\leq \fun b b$ at the $\subv$ position in the typing tree. In the process, the
type $\fun \typ {\fun \typ \typ}$ is updated into $\fun{(\fun b b)}{\fun{(\fun b
    b)}{(\fun b b)}}$. We obtain
$$
\xymatrix@C=0pt@R=2pt{& \app \ar@{-}[dr] \ar@{-}[dl] \\
  \fun{(\fun b b)}{\fun {(\fun b b)}{(\fun b b)}}
  \ar@{-}[d] & & \fun b b \ar@{-}[d] \\
  \mathsf{choose}^{\forall \tvar.{(\fun \tvar {\fun \tvar \tvar})}} & &
  \mathsf{id}^{\forall \tvar.{(\fun \tvar \tvar)}} 
    }
$$
with typing $\term \mathrel : \typing \envvar {\fun{(\fun b b)}{(\fun b b)}}$,
and we can then type $\term ~\term'$. In fact, the expansion mechanism for
subtyping introduction does not depend on the definition of $\leq$, and
therefore we keep \systemfs parametric in its subtyping relation.

\subsection{Summary of contributions}

We define \systemfs and present its principal properties. Improvements over
previous work are as follows:
\begin{enumerate}
\item \systemfs is the first type system with an expansion mechanism for
  $\forall$-quantifiers, where we can delay $\forall$-introduction and uses of
  subtyping with expansion.
\item \systemfs extends the notion of expansion; we introduce a new expansion
  mechanism with its corresponding (asymmetric) E-variables, which differ
  greatly from the ones of System E \cite{CarlierPWK04}.
\item We prove that we can generate all \systemfs judgements from a
  \emph{initial skeleton}, an incomplete typing derivation with constraints that
  need to be solved. This property is a (weaker) form of principality (Theorem
  \ref{t:init-skel}).
\item \systemfs is parametric in its subtyping relation; by using different
  subtyping relations (such as System F type application or Mitchell's relation
  \cite{Mitchell88}), one can change the typing power of \systemfs without
  modifying the typing rules or judgements.
\item \systemfs turns type inference into a type constraint solving
  problem. We believe it can be helpful to reason about modular type inference,
  even if we do not provide a constraint solving algorithm yet. 
\end{enumerate}
The proofs are available in the appendices.

\section{Syntax}
\label{s:syntax}

\begin{figure}[t]
\begin{minipage}[t]{0.26\textwidth}
$$\begin{array}{rllll}
  \expvar & \in & \setexpvar  & \hspace{-2em} ::= & \ev i \\
  \tvar, b & \in & \settvar & \hspace{-2em} ::=& \tv i \\
  \subv & \in & \setsubv & \hspace{-2em} ::= & \sv i \\
  \settv & \in & \finiteset \settvar
\end{array}$$
\end{minipage}
\begin{minipage}[t]{0.73\textwidth}
$$\begin{array}{rllll}
  \term & \in & \setterm & ::= & \expvar \limit \lamb \expvar \term \limit
  \term_1 \app \term_2 \\
  \typ & \in & \settype & ::= & \tvar \limit \fun{\typ_1}{\typ_2} \limit \uquant
  \tvar  \typ \limit \styp \subv \settv \typ \\
  \svar & \in & \setsubs & ::=& \subs {\assign \tvar \typ} \svar \limit \subs{\assign \subv
    \Ivar} \svar \limit \idsubs \\
  \Ivar & \in & \setasymexp &::= & \idi \limit \iquant \tvar \Ivar \limit \iapp
  \subv \settv \Ivar \limit \isub \Ivar \typ \\
  \cons & \in & \setcons &::= & \consomega \limit \typ_1 \consleq
  \typ_2 \limit \cons_1 \consinter \cons_2 \limit \consquant \tvar \cons \limit \conss \subv \settv \typ \cons \\
  \envvar & \in & \settypenv & ::= & \emptyenv \limit \envvar \sepenv \eassign
  \expvar \typ \\
  \skel & \in & \setskel &::= & \skelvar \expvar \envvar \limit \skelabs \expvar
  \skel \limit \skel_1 \skelapp \skel_2 \limit \skelquant \tvar  \skel \limit \skels \subv \settv \skel
  \limit \sktsub \skel \typ
\end{array}
$$
\end{minipage}
\caption{Syntax grammars and metavariable conventions}
\label{f:syntax}
\end{figure}

Fig. \ref{f:syntax} defines the grammars and metavariable conventions of the
entities used in this paper. Let $i$, $j$, $m$, $n$ range over natural
numbers. Given a set $X$, we write $\finiteset X$ for the set of finite subsets
of $X$. We distinguish between the metavariables $\expvar$, $\tvar$, $\subv$,
and the concrete variables $\ev i$, $\tv i$, $\sv i$. The (non-standard) symbol
$\app$ used for application helps in reading skeletons, and we keep it for terms
for consistency. We explain the role of constraints ($\cons$) and skeletons
($\skel$) in Section~\ref{s:typing-rules}, and the syntax of expansion terms
($\Ivar$) and substitutions ($\svar$) in Section \ref{s:subs-exp-application}.

\paragraph*{Precedence.} To reduce parenthesis usage, we define precedence for
operators and operations defined later (such as substitution and expansion
applications $\subst \svar \typ$ and $\Inst \Ivar \settv \typ$) in the following
order, from highest to lowest: $\styp \subv \settv \typ$, $\uquant \tvar 
\typ$, $\subst \svar \typ$, $\Inst \Ivar \settv \typ$,
$\fun{\typ_1}{\typ_2}$. For example, $\fun{\subst \svar {\typ_1}}{\styp \subv
  \settv {\typ_2}}=\fun{(\subst \svar {\typ_1})}{(\styp \subv \settv {\typ_2})}$
and $\fun{\uquant \tvar  \tvar}{\uquant \tvar  \tvar}=\fun{(\uquant
  \tvar  \tvar)}{(\uquant \tvar  \tvar)}$. Furthermore, the function
type constructor is right-associative, so that
$\fun{\typ_1}{\fun{\typ_2}{\typ_3}}=\fun{\typ_1}{(\fun{\typ_2}{\typ_3})}$, and
the application is left-associative, so that $\term_1 \app \term_2 \app \term_3
= (\term_1 \app \term_2) \app \term_3$.

\paragraph*{Equalities and $\alpha$-conversion.} We allow $\alpha$-conversion
of bound variables in types (where $\uquant \tvar \typ$ binds $\tvar$),
skeletons (where $\skelabs \expvar \skel$ binds $\expvar$ and $\skelquant \tvar
\skel$ binds $\tvar$), and constraints (where $\consquant \tvar \cons$ binds
$\tvar$). Note that $\tvar$ is not bound in the expansion term $\iquant \tvar
\Ivar$, and therefore it cannot be $\alpha$-converted. 

We equate types up to reordering of adjacent $\forall$-quantifiers (so $\uquant
{\tvar_1}{\uquant {\tvar_2}{_2}\typ} = \uquant {\tvar_2}{\uquant {\tvar_1}
  \typ}$), and suppression of dummy quantifiers (if $\tvar$ is not free in
$\typ$, then $\uquant \tvar \typ = \typ$). We also enforce the following
equalities on constraints
$$
\begin{array}{rllrllrll}
  \consquant \tvar {(\cons_1 \consinter \cons_2)} &=& (\consquant \tvar
  {\cons_1}) \consinter (\consquant \tvar {\cons_2}) & \quad
  \cons \consinter \cons &=& \cons & \quad \cons \consinter
  \consomega &=& \cons \\
  \conss \subv \settv \typ {(\cons_1 \consinter \cons_2)} &=& (\conss \subv
  \settv \typ {\cons_1}) \consinter (\conss \subv \settv \typ {\cons_2}) &
  \quad \cons_1 \consinter \cons_2 &=& \cons_2 \consinter \cons_1 \\
  \cons_1 \consinter (\cons_2 \consinter \cons_3) &=& (\cons_1 \consinter
  \cons_2) \consinter \cons_3 & \consquant \tvar
  \cons &=& \multicolumn{4}{l}{\cons \mbox { if $\tvar$ is not free in $\cons$}}
\end{array}
$$

\paragraph*{Auxiliary notations and functions.} Let $\fv \term$ be the set of
free variables of $\term$, defined in the usual way. The free type variables of
a type, an expansion, and a substitution are defined as follows.
$$
\begin{array}{llllll}
  \ftv \tvar & = & \{ \tvar \} & \quad \ftv \idi &=& \emptyset \\
  \ftv {\fun{\typ_1}{\typ_2}} & = & \ftv{\typ_1} \cup \ftv{\typ_2} & \quad
  \ftv{\isub \Ivar \typ} &=& \ftv \Ivar \cup \ftv \typ \\
  \ftv {\uquant \tvar  \typ} & = & \ftv \typ \setminus \{ \tvar \} & \quad \ftv
  {\iquant \tvar \Ivar} &=& \ftv \Ivar \cup \{ a \} \\
  \ftv {\styp \subv \settv \typ} & = & \ftv \typ \cup \settv & \quad \ftv {\iapp
    \subv \settv \Ivar} &=& \ftv \Ivar \cup \settv \\
  \multicolumn{6}{c}{
    \begin{array}{lll}
      \ftv \idsubs & =& \emptyset \\
  \ftv{\assign \tvar \typ \sepsubs \svar} &=&\{ \tvar \} \cup \ftv \typ \cup \ftv \svar \\
  \ftv{\assign \subv \Ivar \sepsubs \svar} &=&\ftv \Ivar \cup \ftv \svar
  \end{array}
  }
\end{array}
$$

\section{Typing rules}
\label{s:typing-rules}

\begin{figure}[t]
\begin{mathpar}
  \inferrule
  {}
  {\jtyp {\skelvar \expvar \envvar} \expvar {\envvar}{\apply \envvar \expvar}
    \consomega}~~\TRvar
  \and
  \inferrule
  {\jtyp \skel \term {\eextend \envvar {\eassign \expvar {\typ_1}}} {\typ_2} \cons}
  {\jtyp {\skelabs \expvar \skel}{\lamb \expvar \term} \envvar
    {\fun{\typ_1}{\typ_2}} \cons}~~\TRabs
  \and
  \inferrule
  {\jtyp {\skel_1}{\term_1}{\envvar}{\fun{\typ_1}{\typ_2}}{\cons_1} \\ 
    \jtyp {\skel_2}{\term_2}{\envvar}{\typ_1}{\cons_2} }
  {\jtyp {\skel_1 \skelapp \skel_2}{\term_1 \app \term_2}{\envvar}{\typ_2}{(\cons_1
      \consinter \cons_2)}}~~\TRapp
  \and
  \inferrule
  {\jtyp \skel \term \envvar \typ \cons \\ \tvar \notin \ftv \envvar}
  {\jtyp {\skelquant \tvar  \skel} \term \envvar {\uquant \tvar 
      \typ}{\consquant \tvar \cons}}~~\TRforall
  \and
  \inferrule
  {\jtyp \skel \term \envvar \typ \cons \\ \ftv \envvar \subseteq \settv}
  {\jtyp {\skels \subv \settv \skel} \term \envvar {\styp \subv \settv
      \typ}{\conss \subv \settv \typ \cons}}~~\TRsintro
  \and
  \inferrule
  {\jtyp \skel \term \envvar {\typ_1} \cons}
  {\jtyp {\sktsub \skel {\typ_2}} \term {\envvar}{\typ_2}{(\cons \consinter
      (\typ_1 \consleq \typ_2)})}~~\TRsub
\end{mathpar}
\caption{Typing rules}
\label{f:typing-rules}
\end{figure}

A type environment $\envvar$ (defined in Fig. \ref{f:syntax}) is a list of
assignments which maps term variables to types. When writing a non-empty
environment, we allow omitting the leading symbols ``$\emptyenv$,''.  A type environment is
\emph{well-formed} iff it does not mention twice the same term
variable. Henceforth, we consider only well-formed type environments. For
$\envvar = \envenum {\eassign{\expvar_1}{\typ_1} \sepenv \ldots \sepenv
  \eassign{\expvar_n}{\typ_n}}$, we define
  $\envvar(\expvar_i)=\typ_i$ for $i\in\{1\ldots n\}$,
  $\ftv \envvar = \bigcup_{i \in \{1
  \ldots n\}} \ftv{\typ_i}$, and $\support \envvar = \{\expvar_1 \ldots
\expvar_n \}$. 

The typing rules of \systemfs (Fig. \ref{f:typing-rules}) derive judgements of
the form $\jtyp \skel \term \envvar \typ \cons$, where constraints that need to
be solved (by type inference) are accumulated in $\cons$. A constraint of the
form $\typ_1 \consleq \typ_2$ is called \emph{atomic}. By including constraints
in judgements, we can use the same rules for type checking and type
inference. If the constraint is \emph{solved} w.r.t. some subtyping relation,
then the judgement acts as a regular typing judgement, assigning \emph{typing}
$\typing \envvar \typ$ to the untyped term $\term$. 

A skeleton $\skel$ is just a \emph{proof term}, a compact piece of syntax which
represents a complete typing derivation. A skeleton $\skel$ is \emph{valid} iff
there exist $\term$, $\envvar$, $\typ$, and $\cons$ such that $\jtyp \skel \term
\envvar \typ \cons$. Henceforth, we consider only valid skeletons. All
components of a judgement $\jtyp \skel \term \envvar \typ \cons$ are uniquely
determined by $\skel$, therefore we can define functions $\mathsf{rtype}$ and
$\mathsf{tenv}$ such that $\rtype \skel = \typ$ and $\tenv \skel =
\envvar$. Skeletons replace typing derivation trees in formal statements. For
example, $\skelabs \expvar {\sktsub{(\skelvar \expvar {\envenum{\eassign \expvar
        {\uquant \tvar \tvar}}})}{\fun{(\uquant \tvar \tvar)} b} \skelapp
  \skelvar \expvar {\envenum{\eassign \expvar {\uquant \tvar \tvar}}}}$
represents the following derivation.
\begin{mathpar}
  \inferrule{
    \inferrule*{\jtypnoskel \expvar {\envenum{\eassign \expvar {\uquant \tvar 
            \tvar}}}{\uquant \tvar  \tvar}/\consomega}
    {\jtypnoskel {\expvar}{\envenum{\eassign \expvar {\uquant \tvar 
            \tvar}}} {\fun{(\uquant \tvar 
          \tvar)} b}/(\uquant \tvar \tvar \consleq \fun{(\uquant \tvar \tvar)}
      b)}
    \\
    \inferrule{}
    {\jtypnoskel \expvar {\envenum{\eassign \expvar {\uquant \tvar 
            \tvar}}}{\uquant \tvar  \tvar}/\consomega}
  }
  {\inferrule 
    {\jtypnoskel {\expvar \app \expvar}{\envenum{\eassign \expvar {\uquant \tvar 
            \tvar}}} b/(\uquant \tvar \tvar \consleq \fun{(\uquant \tvar \tvar)}
      b)}
    {\jtypnoskel {\lamb \expvar {\expvar \app \expvar}} \emptyenv {\fun{(\uquant \tvar 
          \tvar)} b}/(\uquant \tvar \tvar \consleq \fun{(\uquant \tvar \tvar)} b)}}
\end{mathpar}
In examples, we sometimes omit skeletons and constraints when they are not
relevant, writing $\term \mathrel : \typing \envvar \typ$ iff there exists
$\skel$, $\cons$ such that $\jtyp \skel \term \envvar \typ \cons$.

\begin{remark}
  A variable skeleton $\skelvar \expvar \envvar$ remembers a type environment
  $\envvar$ and not simply the type of $\expvar$ to be able to type a variable
  $\expvar$ in a term $\lamb \expvar \term$ such that $\expvar \notin \fv
  \term$. For example, we have $\jtyp {\skelabs \expvar {\skelvar y {\eassign
        \expvar a \sepenv \eassign y b}}}{\lamb \expvar y}{\eassign y b}{\fun a
    b} \consomega$.

  We could have used $\uplambda$-terms with only type annotations on
  bindings, like many other systems, but our skeletons are
  also useful because they uniquely represent entire typing derivations (judgement trees).
  We also prefer our skeletons because
  a goal for future work is a system containing both
  System E and \systemfs (cf.\ Section \ref{s:conclusion}), and
  our format of skeleton is better suited for the intersection
  introduction typing rule of System E, as discussed in \cite{WellsH02}.
\end{remark}

Rules \TRvar, \TRabs, and \TRapp are classic. The subtyping rule \TRsub
generates a new atomic constraint, the meaning of which depends on the chosen
subtyping relation (cf. solvedness definition in Section
\ref{s:solvedness}).  Rule \TRforall introduces a $\forall$-quantifier over
$\tvar$, provided that $\tvar$ is not free in $\envvar$. Note that $\tvar$ may
occur free in $\cons$; we use an existential quantifier $\consquant \tvar \cons$
to bind it, as solvedness requires $\cons$ to be solved for some $\tvar$
(cf. Section \ref{s:solvedness}), and not for all possible instantiations of
$\tvar$, as a $\forall$-binder would suggest.

Rule \TRsintro introduces an expansion variable $\subv$ to mark a position in
the derivation tree where a $\forall$-quantifier can be added or where subtyping
can be used. Because a quantification over a free variable of $\envvar$ is not
allowed (rule \TRforall), the E-variable remembers an over-approximation
$\settv$ of $\ftv \envvar$, which is used by the expansion mechanism to prevent
any illegal $\forall$-introduction from happening. The type $\typ$ mentioned in
$\conss \subv \settv \typ \cons$ can be used during expansion to generate an
atomic constraint $\typ \consleq \typ'$ if needed. We explain the expansion
mechanism in detail in the next section.

\begin{remark}
  The rule \TRvar may also introduce E-variables, as for example in $\jtyp{\skelvar
    \expvar {\eassign \expvar {\styp \subv \emptyset \tvar}}} \expvar {\eassign
    \expvar {\styp \subv \emptyset \tvar}}{\styp \subv \emptyset \tvar}
  \consomega$. In this case, performing expansion at the position of $\subv$ does not
  correspond to a use of rules \TRforall or \TRsub, and the set $\settv$ of type
  variables remembered by $\subv$ can be any set. Indeed we can derive
  $\jtyp{\skelvar \expvar {\eassign \expvar {\styp \subv \settv \tvar}}} \expvar
  {\eassign \expvar {\styp \subv \settv \tvar}}{\styp \subv \settv \tvar}
  \consomega$ for any $\settv$.
\end{remark}

\begin{remark}
  In rule \TRsintro, we can remember a set bigger than $\ftv \envvar$ for subject
  reduction to hold. For example, consider the following judgement
  $$\jtyp \skel {(\lamb \expvar y) \app
    \lamb \expvar \expvar}{\eassign y b}{\styp \subv {\{\tvar, b\}} b}
  {\conss \subv {\{\tvar, b\}} b \consomega}$$ 
  with $\skel = (\skelabs \expvar \skels \subv {\{\tvar, b\}}
  {\skelvar y {\eassign \expvar \fun \tvar \tvar \sepenv \eassign y b}}) \app
  \skelabs \expvar {\skelvar \expvar{\eassign \expvar \tvar \sepenv \eassign y
      b}}$. The term $(\lamb \expvar y) \app
  \lamb \expvar \expvar$ reduces to $y$, and to derive 
  $$\jtyp {\skels \subv {\{\tvar, b\}}{\skelvar y {\eassign y b}}} y {\eassign y
    b}{\styp \subv {\{\tvar, b\}} b}{\conss \subv {\{\tvar, b\}} b \consomega},$$ 
  we have to be able to mention $\tvar$ even if it does not appear in $\eassign y b$.
\end{remark}

\begin{figure}[t]
$$
\begin{array}{llllll}
  \Inst \idi \settv \tsk &=& \tsk & \quad \Instc \idi \settv \typ \cons &=& \cons \\
  \Inst {\iapp \subv {\settv'} \Ivar} \settv \tsk & = & \styp \subv {\settv \cup
    \settv'}{(\Inst \Ivar {\settv'} \tsk)} & \quad \Instc {\iapp \subv
    {\settv'} \Ivar} \settv \typ \cons & = & \conss \subv {\settv \cup 
    \settv'}{\Inst \Ivar \settv \typ}{(\Instc \Ivar \settv \typ \cons)} \\
  \Inst {\iquant \tvar \Ivar} \settv \tsk &=& 
  \left\{
    \begin{array}{ll}
      \uquant \tvar  {\Inst \Ivar \settv \tsk} & \mbox{if } \tvar \notin \settv \\
      \Inst \Ivar \settv \tsk & \mbox{otherwise}
    \end{array}
  \right. & \quad \Instc {\iquant \tvar \Ivar} \settv \typ \cons &=& 
  \left\{
    \begin{array}{ll}
      \consquant \tvar {\Instc \Ivar \settv \typ \cons}  & \mbox{if } \tvar \notin \settv \\
      \Instc \Ivar \settv \typ \cons & \mbox{otherwise}
    \end{array}
  \right.\\
  \Inst {\isub  \Ivar {\typ_2}} \settv {\typ_1} &=& \typ_2 & \quad \Instc {\isub
    \Ivar {\typ_2}} \settv {\typ_1} \cons &=& (\Instc \Ivar \settv 
  {\typ_1} \cons) \consinter ((\Inst \Ivar \settv {\typ_1}) \consleq \typ_2) \\
  \Inst {\isub  \Ivar {\typ_2}} \settv \skel &=& \sktsub {(\Inst \Ivar \settv \skel)}{\typ_2}
\end{array}
$$
\caption{Expansion application}
\label{f:exp}
\end{figure}

\section{Substitution and expansion}
\label{s:subs-exp-application}

\subsection{Expansion application}
\label{s:exp-applcation}

The syntax of expansion terms is given in Fig. \ref{f:syntax}. Let $\tsk$ range
over types and skeletons.  Fig. \ref{f:exp} defines the application of expansion
to types, skeletons, and constraints. When applied to a type or a skeleton, the
expansion mechanism relies on a set of type variables $\settv$, used in 
introductions of E-variable and $\forall$-quantifier; when applied to a
constraint, it requires an extra parameter (a type) to generate an appropriate
atomic constraint if needed.  Each construct of expansion terms corresponds to the application
of a non-syntactic typing rule, except for the null expansion $\idi$, which
leaves unchanged the entities it is applied to.

E-variable and $\forall$-quantifier expansions behave the same on types,
skeletons, and constraints. Applied with parameter $\settv$, the expansions
$\iapp \subv {\settv'} \Ivar$ and $\iquant \tvar \Ivar$ first execute $\Ivar$
and then introduce an E-variable $\subv$ (with set $\settv \cup \settv'$ of
variables that cannot be quantified) and a quantifier over $\tvar$ (iff
$\tvar \notin \settv$), respectively. When applied to all parts
of a judgement $\jtyp \skel
\term \envvar \typ \cons$, we must have $\ftv \envvar \subseteq \settv$ for
these operations to be sound w.r.t. rules \TRsintro and \TRforall (cf. Lemma
\ref{l:asym-exp-typings}).

The expansion $\isub \Ivar {\typ_2}$ first applies $\Ivar$ and then performs
subtyping with $\typ_2$, as we can see in the skeleton case. When applied to a
type, only the subtyping step matters, and we simply obtain $\typ_2$. Finally,
the constraint case $\cons$ requires an extra parameter $\typ_1$ to generate a
new atomic constraint. In practice, $\typ_1$ will be the result type of the
judgement $\jtyp \skel \term \envvar {\typ_1} \cons$ from which $\cons$ comes.
When $\isub \Ivar {\typ_2}$ is applied to the above judgement, $\Ivar$ is
applied first, in particular to the type $\typ_1$. To take this into account,
the generated constraint is $(\Inst \Ivar \settv {\typ_1}) \consleq \typ_2$ (and
not simply $\typ_1 \consleq \typ_2$).

Expansion is sound w.r.t. to the type system of \systemfs.
\begin{lemma}
  \label{l:asym-exp-typings}
  If $\jtyp \skel \term \envvar \typ \cons$ and $\ftv \envvar \subseteq \settv$,
  then $\jtyp{\Inst \Ivar \settv \skel} \term {\envvar}{\Inst \Ivar \settv
    \typ}{\Instc \Ivar \settv \typ \cons}$.
\end{lemma}
Expansion operates only at the top-level of the typing judgement in Lemma
\ref{l:asym-exp-typings}; in order to expand at a deeply nested position, we
have to replace an E-variable $\subv$ by an expansion $\Ivar$, as
explained in the next section.

\subsection{Substitution application}
\label{s:subs-application}

\begin{figure}[t]
$$
\begin{array}{llllll}
  \multicolumn{3}{l}{\mbox{\textbf{Metavariables}}} 
  & \quad  \subst \svar {\skelvar \expvar \envvar} & = & \skelvar \expvar {\subst \svar
    \envvar} \\
  \vv ::= \tvar \limit \subv & & 
    & \quad \subst \svar {\skelabs \expvar \skel} &=& \skelabs \expvar {\subst \svar
    \skel}\\
  \tels ::= \typ \limit \Ivar & & & \quad \subst \svar {(\skel_1 \skelapp \skel_2)}
  &=& (\subst \svar {\skel_1}) \skelapp (\subst \svar {\skel_2}) \\
  \multicolumn{3}{l}{\mbox{\textbf{Substitution application}}} & 
  \quad    \subst \svar
  {(\skels \subv \settv \skel)} &=& \Inst {\subst \svar \subv}{\ftv{\subst \svar
      \settv}}{\subst \svar \skel}\\
  \subst \idsubs \tvar & = &\tvar & \quad  \subst \svar {\skelquant \tvar
    \skel} &=& \skelquant {\tvar}{}{\subst \svar \skel} \mbox{ if } \tvar \notin \ftv\svar  \\
  \subst \idsubs \subv & = & \iapp \subv \emptyset \idi & \quad \subst \svar
  {(\sktsub \skel \typ)} &=& \sktsub {\subst \svar \skel}{\subst \svar \typ}\\
  \subst {\subs {\assign \vv \tels} \svar} \vv & = & \tels \\
  \subst {\subs {\assign \vv \tels} \svar} {\vv'} & = & \subst \svar {\vv'} \mbox{ if } \vv
  \neq \vv' & \quad \subst \svar {(\typ_1 \consleq \typ_2)} &=&  \subst \svar {\typ_1} \consleq
  \subst \svar {\typ_2} \\ 
  & & & \quad  \subst \svar \consomega &=& \consomega \\
  \subst \svar {(\styp \subv \settv \typ)} &=& \Inst {\subst \svar
    \subv}{\ftv{\subst \svar \settv}}{\subst
    \svar \typ}  & \quad \subst \svar {(\conss \subv \settv \typ \cons)} &=& \Instc {\subst \svar
    \subv}{\ftv{\subst \svar \settv}}{\subst \svar \typ}{\subst \svar \cons} \\
  \subst \svar {\uquant \tvar  \typ} &=& \uquant \tvar  {\subst \svar
    \typ} \mbox{ if } \tvar \notin \ftv \svar & \quad \subst \svar {\consquant
    \tvar \cons} &=& \consquant \tvar  {\subst \svar
    \cons} \mbox{ if } \tvar \notin \ftv \svar \\
  \subst \svar {(\fun{\typ_1}{\typ_2})} & =& \fun{\subst \svar {\typ_1}}{\subst
    \svar{\typ_2}} & \quad \subst \svar {(\cons_1 \consinter \cons_2)} &=& (\subst \svar {\cons_1})
  \consinter (\subst \svar {\cons_2})
\end{array}
$$
\caption{Substitution application}
\label{f:subs}
\end{figure}

Substitutions (defined in Fig. \ref{f:syntax}) are lists of assignments that map
type variables to types $(\assign \tvar \typ)$ and E-variables to
expansions $(\assign \subv \Ivar)$, ended by the symbol $\idsubs$. Application of
substitutions to type variable sets $\settv$ and type environments $\envvar$ is
pointwise. Given a finite set of types $\{ \typ_1 \ldots \typ_n \}$, we define
$\ftv {\{ \typ_1 \ldots \typ_n \}}$ as $\bigcup_{i \in \{1 \ldots n\}}
\ftv{\typ_i}$. Fig. \ref{f:subs} defines application of substitutions to
variables, types, skeletons, and constraints.

A substitution $\svar$ generates a type $\typ$ (resp. an expansion $\Ivar$) when
applied to a type variable $\tvar$ (resp. to an E-variable $\subv$). A
substitution may contain several assignments for the same variable, as in $\svar
= (\assign \tvar {\typ_1} \sepsubs \assign \tvar {\typ_2} \sepsubs \idsubs)$; in
this case, only the first one is considered. We choose this design for
simplicity; an alternate solution would be to syntactically prevent repetitions
in the substitution definition, but the definition would then become more
complex for no obvious gain.

The application of substitutions to types $\styp \subv \settv \typ$ is the most
important case.
$$\subst \svar {(\styp \subv \settv \typ)} = \Inst {\subst \svar
  \subv}{\ftv{\subst \svar \settv}}{\subst \svar \typ}$$ The substitution
$\svar$ is first applied to $\subv$, which gives us an expansion $\Ivar=\subst
\svar \subv$, which is then applied to the type $\subst \svar \typ$. We remember
that $\settv$ is (an over-approximation of) the set of free type variables that
cannot be quantified over, because they appear in the type environment
at the time the variable $\subv$ is introduced. If $\svar$ replaces a variable
$\tvar \in \settv$ by a type $\typ'$, then $\typ'$ now appears in the type
environment, and its free variables cannot be quantified over. This explains why
we have to apply the expansion $\Inst {\subst \svar \subv}{\ftv{\subst \svar
    \settv}}{\subst \svar \typ}$ with the set $\ftv {\subst \svar \settv}$ and
not simply with the set $\settv$. The application of $\svar$ to skeletons
$\skels \subv \settv \skel$ and to constraints $\conss \subv \settv \typ \cons$
follows the same pattern.

\begin{example}
  Let $\term = \lamb \expvar {\expvar \app y}$. We have
  $$\jtypnoskel {\term}{\eassign y \tvar}
  {\styp \subv {\{\tvar \}} {(\fun{(\fun \tvar b)} b)}}$$
  Applying $\svar_1 = (\assign \tvar {\fun{\tvar_1}{\tvar_2}} \sepsubs \idsubs)$ 
  to this typing gives us 
  $$
    \jtypnoskel{\term} {\eassign y {\fun
        {\tvar_1}{\tvar_2}}}
  {\styp \subv {\{\tvar_1, \tvar_2 \}} {(\fun{(\fun
        {(\fun{\tvar_1}{\tvar_2})}b)} b)}}
  $$
  Then applying $\svar_2 = (\assign \subv {\iquant b \idi} \sepsubs \idsubs)$ gives us
  $$
  \jtypnoskel{\term}{\eassign y
      {\fun {\tvar_1}{\tvar_2}}} {\uquant b  {(\fun{(\fun
        {(\fun{\tvar_1}{\tvar_2})}b)} b)}}
  $$
  Note that the substitution $(\assign \subv {\iquant {\tvar'} \idi} \sepsubs
  \idsubs)$ would have left the last judgement unchanged if $\tvar' \in \{
  \tvar_1, \tvar_2\}$, and would have introduced a dummy quantifier if $\tvar'
  \notin \{ b, \tvar_1, \tvar_2 \}$. We can achieve the same effect as doing
  $\svar_1$ before $\svar_2$ by applying the substitution $\svar = (\assign
  \tvar {\fun{\tvar_1}{\tvar_2}} \sepsubs \assign \subv {\iquant b \idi}
  \sepsubs \idsubs)$ to the initial judgement.
\end{example}

\begin{example}
  Let $\typ = \uquant \tvar  {(\fun \tvar \tvar)}$. We have
  $$
  \jtyp
  {\skelabs \expvar {\skels \subv \emptyset {(\sktsub {(\skelvar
        \expvar {\eassign \expvar \typ})}{\fun{\typ}{\typ}} \app {\skelvar \expvar
           {\eassign \expvar \typ}})}}}
  {\lamb \expvar {\expvar \app \expvar}}{\emptyenv}
  {\fun \typ {\styp \subv \emptyset \typ}}
  {\conss \subv \emptyset \typ {(\typ \consleq \fun \typ \typ)}}
  $$
  Applying substitution $\svar = (\assign \subv {\isub \idi {\fun b b}} \sepsubs \idsubs)$ gives us
  $$\jtyp
  {\skelabs \expvar {\sktsub {(\sktsub {(\skelvar
        \expvar {\eassign \expvar \typ})}{\fun \typ \typ} \app {\skelvar \expvar
        {\eassign \expvar \typ}})}{\fun b b}}}
  {\lamb \expvar {\expvar \app \expvar}}{\emptyenv}
  {\fun{\typ}{\fun b b}}
  \cons
  $$
  where $\cons = (\typ \consleq \fun b b) \consinter (\typ \consleq \fun \typ \typ)$.
  Subtyping has been introduced at a nested position (under the
  $\lambda$), generating the expected constraint $\typ \consleq \fun b b$.
\end{example}

Substituting a variable $\subv$ by an expansion $\Ivar$ makes $\subv$
disappear. As a result, one can use the null expansion $\idi$ to delete an
E-variable $\subv$ from a type $\styp \subv \settv \typ$. If $\svar = (\assign
\subv \idi \sepsubs \idsubs)$, then $\subst \svar {(\styp \subv \settv \typ)}
=\Inst \idi \settv {\subst \svar \typ}=\subst \svar \typ$ (the occurrences of
$\subv$ in $\typ$ are also removed).
An expansion $\Ivar$ can be applied at the location of a variable $\subv$
without making $\subv$ disappear using the
substitution $\svar=(\assign \subv {\iapp \subv \emptyset \Ivar} \sepsubs
\idsubs)$. Indeed we have $ \subst \svar {(\styp \subv \settv \typ)} =\Inst
{\iapp \subv \emptyset \Ivar} \settv {\subst{\svar}{\typ}} =\styp \subv \settv
{\Inst \Ivar \settv {\subst{\svar}{\typ}}}$. The substitution $\idsubs$ is the
\emph{identity} substitution; it leaves variables, types, skeletons, and
constraints unchanged. For example, for E-variables, we have $\subst \idsubs
{(\styp \subv \settv \typ)} = \Inst {\iapp \subv \emptyset \idi} \settv {\subst
  \idsubs \typ} = \styp \subv \settv {\subst \idsubs \typ}$. The remaining cases
of substitution application are straightforward descending cases. The resulting
operation is sound w.r.t. \systemfs type system.

\begin{theorem}
  \label{t:subs-admissible}
  If $\jtyp \skel \term \envvar \typ \cons$ then $\jtyp{\subst \svar \skel}
  \term {\subst \svar \envvar}{\subst \svar \typ}{\subst \svar \cons}$.
\end{theorem}

\section{Initial Skeletons}
\label{s:principality}

\begin{figure}[t]
\begin{mathpar}
  \inferrule{}
  {\jinit \xtoa \expvar {\skels \subv {\ftv \xtoa}{\skelvar \expvar \xtoa}}}
  \and
  \inferrule
  {\jinit{\xtoa \sepxtoa \eassign \expvar \tvar} \term \skel \\ \subv \notin
    \allvar \skel \\ \settv = \ftv{\tenv{\skelabs \expvar \skel}}}
  {\jinit \xtoa  {\lamb \expvar \term}{\skels \subv \settv {(\lamb \expvar
        \skel)}}}
  \and
  \inferrule
  {\jinit \xtoa {\term_1}{\skel_1} \\ \jinit \xtoa {\term_2}{\skel_2} \\ \skel =
    \sktsub {\skel_1}{\fun{\rtype{\skel_2}} \tvar} \skelapp \skel_2 \\
    \settv = \ftv{\tenv \skel} \\\\
    (\allvar{\skel_1} \cap \allvar{\skel_2}) \setminus \ftv \xtoa = \emptyset \\
    \{ \tvar, \subv \} \cap (\allvar {\skel_1} \cup \allvar{\skel_2}) = \emptyset }
  {\jinit \xtoa  {\term_1 \app \term_2}{\skels \subv \settv \skel}}
  \and
  \inferrule
  {\jinit \xtoa \term \skel \\ \support \xtoa = \fv \term}
  {\jinitsk \term \skel}
\end{mathpar}
\caption{Initial skeletons of a term}
\label{f:init-skel}
\end{figure}

In this section, we prove that we can generate all \systemfs judgements for a
term $\term$ from an initial skeleton built from $\term$.

We first show that we can obtain \emph{relevant} skeletons; a skeleton $\skel$
such that $\jtyp \skel \term \envvar \typ \cons$ is relevant if $\fv \term =
\support \envvar$. In words, the type environment of a relevant skeleton does
not mention more term variables than necessary. A \emph{variable environment}
$\xtoa$ is a type environment which assigns type variables to expression
variables and such that for all $\expvar$, $y$ such that $\expvar \neq y$, we
have $\apply \xtoa \expvar \neq \apply \xtoa y$. We write $\allvar \skel$ for
the set of free type and E-variables occurring in $\skel$. Fig. \ref{f:init-skel}
defines a judgement $\jinitsk \term \skel$, which means that $\skel$ is an
\emph{initial skeleton} for $\term$. The main ideas behind this construct are as
follows: first, we type each variable in $\fv \term$ with a distinct type
variable (using the environment $\xtoa$ mentioned in the auxiliary judgement $\jinit \xtoa
\term \skel$). Then we introduce a (fresh) E-variable at every possible position
in the skeleton. Finally, we use subtyping to ensure that a term in a function
position in an application has an arrow type. Two initial skeletons for the same
term are equivalent up to renaming of their variables, as stated in the lemma
below (where we call an expansion of the form $\iapp \subv
\settv \idi$ an \emph{E-expansion}).

\begin{lemma}
  Let $\skel_1$, $\skel_2$ such that $\jinitsk \term {\skel_1}$ and $\jinitsk
  \term {\skel_2}$. There exists a substitution $\svar$ which maps type
  variables to type variables and E-variables to E-expansions such that $\skel_1
  = \subst \svar {\skel_2}$.
\end{lemma}

\begin{example}
  \label{ex:init-1}
  Let $\term = \lamb \expvar {\expvar \app \expvar}$. Then
  $$\skel = 
  \skels {\sv 3}{ \emptyset }{
    \skelabs \expvar {
      \skels {\sv 2}{ \{ \tv 0 \} }{(
        \sktsub {(\skels {\sv 0}{ \{\tv0 \} }{
            \skelvar \expvar {\eassign \expvar {\tv 0}}}
          )}{(\fun {\styp {\sv 1}{ \{\tv 0\} }{\tv 0}}{\tv 1})}
        \app
        \skels {\sv 1}{ \{\tv 0 \} }{
          \skelvar \expvar {\eassign \expvar {\tv 0}}}
        )}
    }
  }$$
  is an initial skeleton for $\term$ and we have
  $$\jtyp \skel \term {\emptyenv}{\styp
    {\sv 3} \emptyset {(\fun {\tv 0}{\styp {\sv 2}{ \{\tv 0 \}}{\tv
          1}})}}{\cons}$$
  with $\cons = \conss {\sv 3} \emptyset {\fun {\tv 0}{\styp {\sv 2}{ \{\tv 0 \}}{\tv
        1}}}{\conss {\sv 2}{\{ \tv 0 \}}{\tv 1}{((\styp {\sv 0}{ \{ \tv 0 \}
      }{\tv 0} \consleq \fun {\styp {\sv 1}{ \{\tv 0\} }{\tv 0}}{\tv 1})
      \consinter \conss {\sv 0}{\{ \tv 0 \} }{\tv 0} \consomega \consinter
      \conss {\sv 1}{\{ \tv 0 \} }{\tv 0} \consomega )}}$.  
  Roughly, the variables ($\sv i$) can be used to introduce
  $\forall$-quantifiers or subtyping at their respective positions. 
  For example, let $\typ
  = \uquant \tvar  {(\fun \tvar \tvar)}$ and $\svar=(\assign {\tv 0} \typ \sepsubs
  \assign {\tv 1} \typ \sepsubs \assign {\sv
    0} \idi \sepsubs \assign {\sv 1} \idi \sepsubs \assign {\sv 2}{\isub \idi
    {\fun b b}} \sepsubs \assign {\sv 3}{\iquant b \idi} \sepsubs
  \idsubs)$. Applying $\svar$ to the above typing judgement, we obtain
  $$\jtyp {\skelquant b  {\skelabs \expvar {\sktsub {(
          \sktsub {(\skelvar \expvar {\eassign \expvar \typ})}{\fun \typ \typ}
          \skelapp \skelvar \expvar {\eassign \expvar \typ})}{\fun b b}}}} \term
  {\emptyenv}{\uquant b  {(\fun {\typ}{\fun b b})}}{\subst \svar \cons}$$ 
  with 
  $
  \subst \svar \cons = \consquant b {((\typ \consleq \fun \typ \typ) \consinter
    (\typ \consleq \fun b b))}
  $.
  \qed
\end{example}

In the following, we use a predicate $\rawrefl$ to check that a constraint is
built from atomic constraints of the form $\typ \consleq \typ$. The formal
definition is
\begin{mathpar}
  \refl \consomega 
  \and
  \refl{\typ \consleq \typ}
  \and
  \inferrule{\refl \cons}
  {\refl {\consquant \tvar \cons}}
  \and
  \inferrule{\refl \cons}
  {\refl {\conss \subv \settv \typ \cons}}
  \and 
  \inferrule{\refl {\cons_1} \\ \refl{\cons_2}}
  {\refl {\cons_1 \consinter \cons_2}}
\end{mathpar}
A reflexive constraint is always solved w.r.t. a reflexive subtyping relation
(see solvedness definition in the next section). From any initial skeleton of
$\term$, we can obtain all relevant skeletons for $\term$.
\begin{lemma}
  \label{l:init-skel}
  Let $\jinitsk \term \skel$. Let $\skel'$ relevant such that $\jtyp {\skel'}
  \term \envvar \typ \cons$. There exists $\svar$ such that $\jtyp{\subst \svar
    \skel} \term {\envvar}{\typ}{(\cons \consinter \cons')}$ with $\cons'$
  reflexive.
\end{lemma}
Note that in the above lemma, we do not have $\subst \svar \skel = \skel'$, and
we obtain an approximation of $\cons$. By construction, an initial skeleton
$\skel$ uses subtyping at each application node to generate an atomic
constraint. Applying $\svar$ turns these constraints into reflexive ones, but
it cannot completely remove them. Therefore, $\subst \svar \skel$ is similar to
$\skel'$ up to these uses of (reflexive) subtyping at application nodes.

To generate all possible typing derivations, we add a weakening rule to be able
to extend a type environment.
\begin{mathpar}
  \inferrule{\jtyp \skel \term {\envvar_1} \typ \cons \\ \support{\envvar_1} \cap
    \support{\envvar_2} = \emptyset}
  {\jtyp {\skelweak \skel {\envvar_2}} \term {\envvar_1 \sepenv \envvar_2} \typ
    \cons}
\end{mathpar}
\begin{theorem}
  \label{t:init-skel}
  Let $\jinitsk \term \skel$. If $\jtyp {\skel'} \term \envvar \typ \cons$, then
  there exists $\svar$, $\envvar'$ such that $\jtyp{\skelweak {(\subst \svar
      \skel)}{\envvar'}} \term {\envvar}{\typ} {(\cons \consinter \cons')}$,
  with $\cons'$ reflexive.
\end{theorem}
We emphasize that initial skeletons are quite different from principal typings:
initial skeletons are not typing derivations, because they contain unsolved
constraints, and all terms, even non typable ones, have an initial skeleton. To
obtain a principal typing from the initial skeleton, we need to solve the
constraints in a principal manner; we conjecture that it is not possible, i.e.,
\systemfs does not have principal typings, for the same reason as for System F
\cite{Wells02}.

Nevertheless, we think that initial skeletons can be useful for modular type
inference. First, note that we do not have to remember the skeleton itself or
the term; the typing and constraint contain all the information we
need. Besides, constraint solving can be divided into solution preserving
steps, which produce an equivalent constraint, and solution reducing steps,
where some information is lost. It is always possible to safely perform solution
preserving steps, and one can periodically check if it is possible to apply
solution reducing steps to find at least one solved typing. The best
intermediate representation might be a typing on which all known solution
preserving steps have been performed, together with (at least) one solution
reducing step of that typing's constraint. We do not know in practice how many
steps will be solution preserving versus solution reducing.

An example use of \systemfs is to look for a subsystem of System F in which to
do compositional type inference. \systemfs is a good framework in which to
perform such a search, by considering various different restrictions of
\systemfs until one is found with the right properties.  Because all possible
System F derivations can be obtained from \systemfs initial skeletons, we know
in advance that the framework has the right amount of power.  Such subsystems
could also be characterized by a constraint solving algorithm. Instead of
searching for a subsystem by varying the typing rules, we could vary the constraint
solving algorithm, and when a nice algorithm is found, we could try to
find a corresponding restriction directly stated on the typing rules.

\section{Solvedness and Subject Reduction}

\subsection{Solvedness and System F}
\label{s:solvedness}

A constraint $\cons$ is solved w.r.t. a subtyping relation $\leq$ if its atomic
constraints are solved w.r.t. $\leq$. Formally, we define the predicate
$\rawsolved$, as follows.
\begin{mathpar}
  \solved \consomega \leq
  \and
  \inferrule{\typ_1 \leq \typ_2}
  {\solved {\typ_1 \consleq \typ_2} \leq}
  \and
  \inferrule{\solved {\cons_1} \leq \\ \solved{\cons_2} \leq }
  {\solved {\cons_1 \consinter \cons_2} \leq}
  \and
  \inferrule{\solved \cons \leq}
  {\solved {\consquant \tvar \cons} \leq}
  \and
  \inferrule{\solved \cons \leq}
  {\solved {\conss \subv \settv \typ \cons} \leq}
\end{mathpar}
A skeleton is solved if its constraint is solved. Solved skeletons correspond
to typing derivations in the traditional sense.

We can express System F in \systemfs by using the following relation $\leqf$.
\begin{mathpar}
  \inferrule
  {}
  {\uquant \tvar  {\typ_1} \leqf \subst {\assign \tvar
      {\typ_2} \sepsubs \idsubs}{\typ_1}}~~\SRelim
\end{mathpar}
Because of the equality involving dummy quantifiers, the relation $\leq$ is
reflexive; indeed for $\tvar \notin \ftv \typ$, we have $\typ = \uquant \tvar
\typ \leqf \typ$. Clearly, \systemfs equipped with $\leqf$ extends System
F. Conversely, it is easy to see that a term typable in \systemfs is typable in
F once we erase all the E-variables.

\begin{proposition}
  A term is typable in System F iff it is typable in \systemfs with $\leqf$. 
\end{proposition}

\subsection{Subject Reduction}

We now present the subject reduction result of \systemfs with $\leqf$
w.r.t. call-by-value semantics. Let $\vvar$ range over values, i.e. $\vvar ::=
\expvar \limit \lamb \expvar \term$. We write $\subst{\assign \expvar
  {\term_1}}{\term_2}$ for the usual capture-avoiding substitution of terms. We
define small-step call-by-value evaluation $\term \cbv \term'$ as the smallest
relation on terms verifying the following rules:
\begin{mathpar}
  \inferrule
  {}
  {(\lamb \expvar \term) \app \vvar \cbv \subst {\assign \expvar \vvar} \term}
  \and
  \inferrule
  {\term_1 \cbv \term'_1}
  {\term_1 \app \term_2 \cbv \term'_1 \app \term_2}
  \and
  \inferrule
  {\term \cbv \term'}
  {\vvar \app \term \cbv \vvar \app \term'}
\end{mathpar}
\begin{theorem}
  \label{t:subject-reduction}
  If $\jtyp \skel \term \envvar \typ \cons$, $\solved \cons \leqf$, and $\term
  \cbv \term'$, then there exists $\skel'$, $\cons'$ such that $\jtyp
  {\skel'}{\term'} {\envvar} \typ {\cons'}$ and $\solved {\cons'} \leqf$.
\end{theorem}
We prove Theorem \ref{t:subject-reduction} by defining a transformation on
$\skel$ so that skeletons in a function position of an application, such as
$\skel_1$ in $\skel_1 \skelapp \skel_2$, are turned into $\uplambda$-abstraction
skeletons. A substitution lemma then allows us to simulate $\upbeta$-reduction
by replacing the occurrences of a variable skeleton $\skelvar \expvar \envvar$
in a skeleton $\skelabs \expvar {\skel'_1}$ by $\skel_2$. This proof technique
depends on the subtyping relation being used. We conjecture it can be adapted to
various relations (such as Mitchell's \cite{Mitchell88}), but nevertheless we
look for a more generic proof technique (less dependant on the subtyping
relation). We prove subject reduction only for call-by-value evaluation for
simplicity; we conjecture that subject reduction also holds for call-by-need and
call-by-name semantics, and for reduction in arbitrary contexts.

\section{Related Work}
\label{s:related}

\subsection{Expansion}

A full survey on expansion and expansion variables can be found in
\cite{CarlierW05}; we only discuss here the main differences between \systemfs
and System E, the type system with expansion most closely related to our
work. System E E-variables are introduced on top of skeletons, type
environments, result types, and constraints, while \systemfs E-variables are not
inserted on top of type environments (rule \TRsintro). \systemfs expansion
mechanism deals with subtyping, while System E expansion does not. In System E,
an E-variable $\evar$ defines a namespace. In type $\typ_1 = \fun \tvar {\etyp
  \evar \tvar}$, the variable $\tvar$ outside $\evar$ is not connected to the
one in the scope of $\evar$; applying substitution $(\assign \tvar {\typ_2}
\sepsubs \idsubs)$ to $\typ_1$ gives $\fun {\typ_2}{\etyp \evar \tvar}$. This is
due to the fact that substitutions are a special case of System E expansions
(see \cite{CarlierW05} for further details). It also makes composition of
expansions and substitutions easier. In \systemfs, substitutions cannot be
considered as expansions, because they are applied to the whole typing judgement
(Theorem \ref{t:subs-admissible}), whereas the asymmetric expansions of
\systemfs are not applied to the type environments (Lemma
\ref{l:asym-exp-typings}). As a result, it would be unsound for \systemfs
E-variables to create namespaces. It is difficult to have a symmetric expansion
in \systemfs, because subtyping does not operate uniformly on typings (it is
usually contravariant on the environment and covariant on the result type). It
is possible to design \systemfs with two kinds of E-variables (one, symmetric,
to handle substitutions and $\forall$-introduction, and one, asymmetric, for
subtyping), but it would make the system much more complex for no clear profit.

\subsection{Type Inference in System F}

Type inference in System F is undecidable \cite{Wells99}; however many different
approaches have been conducted to circumvent this issue, by stratifying System F
using a notion of rank, or by using type annotations to constrain type inference
possibilities.

\paragraph*{Giannini and Ronchi's type constraints.} In \cite{GianniniR91},
Giannini and Ronchi Della Rocca consider a syntax-directed version of System
F. The authors define a notion of \emph{typing scheme} $\sigma$, with a syntax
similar to the one of System F types, except that quantifiers $\forall u.\typ$
contain placeholders $u$ (called \emph{sequence variables}), that can be
replaced by a (possibly empty) set of type variables to give a System F
type. For each term $\term$, they also define a \emph{principal typing scheme}
$\apply \Uppi \term = \langle D, \sigma, G, F \rangle$, where $D$ is an
environment that maps term variables to typing schemes, and $G$ and $F$ are
constraints on the typing schemes occurring in $\sigma$ or $D$ that need to be
satisfied. The set $G$ contains subtyping constraints $\sigma_1 \leqf \sigma_2$,
and $F$ prevents certain quantifications from happening by restricting the
possible values for the sequence variables $u$.

The principal typing scheme $\apply \Uppi \term$ is similar to our initial
skeletons; if $\apply \Uppi \term = \langle D, \sigma, G, F \rangle$ and $\jtyp
\skel \term \envvar \typ \cons$ (with $\skel$ an initial skeleton for $\term$),
then $D$ corresponds to $\envvar$, $\sigma$ to $\typ$, $G$ to $\cons$, and $F$
acts as the sets $\settv$ that appear in E-variables $\styp \subv
\settv \typ$. Any System F typing $\typing \envvar \typ$ of $\term$ can be
obtained from $D$, $\sigma$ by applying a substitution (from type variables to
types and sequence variables to set of type variables) which satisfies
constraints $G$ and $F$. This result corresponds to Theorem \ref{t:init-skel} in
our system.

\systemfs and the system of \cite{GianniniR91} differ mainly in their
implementation. In particular, we have a mechanism to postpone subtyping (i.e.,
$\forall$-elimination), which does not have an equivalent in the system of
Giannini and Ronchi. It seems that they do not need such mechanism, but to
compensate for it, they have to generate more constraints when building their
principal typing scheme $\apply \Uppi \term$. We also believe that our system is
easier to understand and easier to extend with other type constructors. Finally,
Giannini and Ronchi define a notion of rank over system F types (distinct from
Leivant's rank based on the presence of polymorphism on the left of function
types \cite{Leivant83}), and provide for all $n$ an inference algorithm for each
restriction of their system to types of rank lower than $n$. We conjecture that
this algorithm can be adapted to \systemfs.

\paragraph*{\mlf and its variants.} \mlf \cite{BotlanR03,BotlanR09} is a
conservative extension of ML at least as expressive as System F with principal
types, i.e., result types whose instances (w.r.t. the \mlf type instance
relation $\leqmlf$) are exactly all possible result types for a term. The type
system also enjoys decidable type inference (with a simple criterion on where
type annotations are needed), and stability w.r.t. some program transformations,
such as for example $\upbeta$-reduction and $\upeta$-expansion. 

\mlf types contain \emph{flexible quantifiers} $\forall(\tvar \geqmlf
\sigma)\sigma'$, which roughly represent sets of System F types of the form
$\subst {\assign \tvar \typ}{\typ'}$, where $\typ$ and $\typ'$ are instances of
the \emph{type schemes} $\sigma$, $\sigma'$. For example, $\forall(\tvar \geqmlf
\forall b (\fun b b))(\fun \tvar \tvar)$ represents the set $\{ \fun \typ \typ
\sthat \forall b (\fun b b) \leqmlf \typ \}$. With flexible quantifiers, terms
that do not have a principal type in System F (w.r.t. the System F type instance
relation) have a principal type in \mlf. Decidable type inference is obtained in
\mlf by requiring type annotations on function parameters that are used two or
more times with different type instances, so that the type inference algorithm
never has to guess true polymorphism. \emph{Rigid bindings} are used in \mlf
types and typing rules to distinguish between inferred and annotated types. They
are not necessary for decidable type inference, and can be removed at the cost
of additional type annotations, as in HML \cite{Leijen09}.

\paragraph*{Boxed polymorphism.} Boxed polymorphism \cite{LauferO94,Remy94}
hides polymorphic types into boxes, considered as regular simple types. Several type
systems follow this principle, such as PolyML \cite{GarrigueR99}, boxy types
\cite{VytiniotisWJ06}, and FPH \cite{VytiniotisWJ08}. We discuss only the most
recent system, FPH. FPH is a type system based on System F, where boxes are used
to mark where $\forall$-quantifiers have to be instantiated with polymorphic
types. Provided that type annotations are given at these boxed positions, FPH
type inference computes System F types (without any box) for terms. The system
aims for simplicity for the programmer: only System F types are exposed, and
writing type annotations does not require to think in term of boxes. Roughly,
type annotations are necessary for $\lambda$-abstractions and let-bindings with
\emph{rich} types (i.e., types with quantifiers under arrow types). However, FPH
is more restrictive than \mlf; more annotations are needed in general, and FPH
terms admit principal types only for ``box-free'' types, not in general.\\

\mlf, FPH, and \systemfs all aim for a modular type inference for System F
types. It is difficult to compare our work to these two systems, because we do
not propose a type inference algorithm for \systemfs yet. In particular,
assuming we follow their approach, we do not know how many annotations would be
necessary to make \systemfs type inference decidable. However, we can make the
following observations. First, \mlf and FPH only infer result types, while our
objective is to also infer complete typing, in order to have a fully
compositional type inference algorithm. \mlf has principal types (w.r.t. to
their instance relation), while \systemfs have initial skeletons, and FPH has
principal types only for box-free types (where $\forall$-quantified variables
cannot be instantiated with polymorphic types). \mlf types more terms than System
F, while FPH and \systemfs type the same terms as System F. Finally, FPH and
\systemfs are direct extensions of System F, and the constructions specific to
these systems (the boxes and E-variables) can be kept away from the programmer
most of the time (except in type error reports). On the other hand, \mlf types
and type instance relation $\leqmlf$ can be hard to understand, even in its
simpler version HML.

To illustrate the differences between the three type systems, we consider the
following example (taken from \cite{BotlanR09,VytiniotisWJ08}). Let $\envvar =
\envenum{\eassign {\mathsf{choose}}{\uquant \tvar  {(\fun \tvar {\fun \tvar
        \tvar} )}} \sepsubs \eassign{\mathsf{id}}{\uquant \tvar  {(\fun
      \tvar \tvar)}}}$ and $\term = \mathsf{choose} \app \mathsf{id}$. We can
derive the following typing judgement for $\term$:
$$
\begin{array}{rl}
  \mbox{F$_{\mathsf s}$}:& \typing \envvar {\styp {\subv_2} \emptyset {(\fun {(\styp {\subv_1}
        \emptyset {\uquant \tvar  {(\fun \tvar \tvar)}})}{(\styp {\subv_1}
        \emptyset {\uquant \tvar  {(\fun \tvar \tvar)}})})}}\\
  \mbox{\mlf}:& \typing \envvar {\forall(\tvar \geqmlf \forall b (\fun b
  b))(\fun \tvar \tvar)} \\
  \mbox{FPH}:& \typing \envvar {\forall b (\fun {(\fun b b)}{(\fun b b)})} \\
  & \typing \envvar {\fun {\fbox{$\forall b (\fun b b)$}}{\fbox{$\forall
        b (\fun b b)$}}}
\end{array}
$$
FPH can infer two result types for $\term$, depending on the presence or absence
of type annotations. These two incomparable types can be obtained from the
(principal) \mlf type (ignoring the boxes), and also from the \systemfs type, by
applying the substitution $(\assign {\subv_2}{\iquant b \idi} \sepsubs \assign
{\subv_1}{\isub \idi {\fun b b}} \sepsubs \idsubs)$ for the first one, and by
simply erasing the E-variables for the second one.

Both \systemfs E-variables and \mlf flexible bindings factor several System~F
types and typing derivations that are incomparable in System F, as shown with
the $\mathsf{choose} \app \mathsf{id}$ example. However, flexible bindings are
more expressive and allow to type terms that are not typable in System
F. Consider the example (taken from \cite{BotlanR09}) $\mathsf{let~} x =
(\mathsf{choose} \app \mathsf{id}) \mathsf{~in~let~} z = x \app f \mathsf{~in~}
x \app g$, where $\eassign f {\fun{\uquant \tvar  {(\fun \tvar
      \tvar)}}{\uquant \tvar  {(\fun \tvar \tvar)}}} \sepenv \eassign g
{\fun {(\fun b b)}{(\fun b b)}}$. The \mlf type for $\mathsf{choose} \app
\mathsf{id}$ given above can be instantiated into the incomparable types of $f$
and $g$. The term cannot be typed in System F nor in \systemfs. Adding quantification
over E-variables would allow \systemfs to type this term; we could type
$\mathsf{choose} \app \mathsf{id}$ with $\forall \subv.{(\fun {(\styp \subv
    \emptyset {\uquant \tvar  {(\fun \tvar \tvar)}})}{(\styp \subv \emptyset
    {\uquant \tvar  {(\fun \tvar \tvar)}})})}$ and instantiate $\subv$ with
different expansions to obtain the types of $f$ and $g$. Adding
quantification over
E-variables should not raise any issue; we conjecture that it would allow
\systemfs to type as many terms as \mlf. It would be interesting to see if there
exists an encoding of \mlf types into \systemfs types extended with
quantified E-variables, and
conversely. We leave this topic to future work.

\section{Conclusion and Future Work}
\label{s:conclusion}

\systemfs is an extension of System F with expansion, an operation originally
defined in systems with intersection types. Expansion allows postponing the
introduction of $\forall$-quantifiers and subtyping uses at an arbitrary nested
position in a typing derivation. For any term $\term$, we can generate an
initial skeleton, from which we can obtain any \systemfs judgement for
$\term$. We now give some ideas of follow-up on this work.

\paragraph*{Type inference algorithm.} To obtain decidable type inference in
\systemfs, a first possibility is to use type annotations, as in \mlf or
FPH. The question is then to know how many annotations are necessary
compared to these two systems. Another idea is to study the link between
constraints solving and semi-unification. Given a constraint $\typ_1 \leq
\typ_2$, the semi-unification problem consists in finding $\svar_1$, $\svar_2$
so that $\subst {\svar_2}{\subst {\svar_1}{\typ_1}} =
\subst{\svar_1}{\typ_2}$. Vasconcellos et al.~\cite{VasconcellosFC03} used
semi-unification to design and implement a type inference semi-algorithm for
polymorphic recursion in Haskell. The authors claim that the algorithm
terminates most of the time in practice. Maybe similar results can be obtained
for \systemfs as well. As discussed at the end of Section \ref{s:principality},
\systemfs can also be used to look for a subsystem of System F allowing for
compositional type inference.

\paragraph*{Mixing $\forall$-quantifiers and intersection types.} A long-term
goal is combining System E and System F into one system (called System EF), with
both $\forall$-quantifiers and intersection types. With such a system, one could
type a term with only intersection types, only System F types, or any
combination of the two constructs, depending on the user's needs. Previous
systems featuring both constructs (e.g. \cite{MargariaZ95,BakelBF99}) do not
use expansion variables; the main difficulty in mixing System E and
\systemfs is to make precise the interactions between the symmetric and
asymmetric expansions. Maybe it is possible to define a more general expansion
mechanism which supersedes the existing ones, and combine the two kinds of
expansion variables into a single construct.  A goal would be for System EF to
have principal typings.

Because System E types all strongly normalizing terms,
$\forall$-quantified types would only be used when required by the
user when performing type inference in
System EF. To this end, we could imagine various kinds of type annotations to
mark positions within terms where System F types are required. These annotations
could be complete types, such as $\lamb {\expvar^{\uquant \tvar  {(\fun
      \tvar \tvar)}}} \term$, or just type templates, such as $\lamb
{\expvar^{\fun {(\forall \tvar.*)} *}} \term$, meaning that the inferred type
for $\expvar$ should be an arrow type, and the type of the argument should be a
System F type. One could imagine different kinds of annotations at various
positions in the term; we would like to see under which conditions (on both the
annotations language and the positions in the term) the inference for such a
system becomes decidable. The inference algorithm would then use intersection
types by default, except for the marked positions where $\forall$-quantified
types are requested.

\bibliography{biblio}
\bibliographystyle{abbrv}

\newpage
\appendix

\section{Soundness of substitutions}

\begin{lemma}
  \label{l:iinst-preserves-typings}
  If $\jtyp \skel \term \envvar \typ \cons$ and $\ftv \envvar \subseteq \settv$
  then we have $\jtyp{\Inst \Ivar \settv \skel} \term {\envvar}{\Inst \Ivar \settv
    \typ}{\Instc \Ivar \settv \typ \cons}$.
\end{lemma}

\begin{proof}
  By induction on $\Ivar$.

  If $\Ivar = \idi$, then the result is easy. 

  If $\Ivar = \isub {\Ivar'}{\typ_2}$, then by induction we have $\jtyp{\Inst
    {\Ivar'} \settv \skel} \term {\envvar}{\Inst {\Ivar'} \settv \typ}{\Instc
    {\Ivar'} \settv \typ \cons}$. By rule \TRsub, we have $\jtyp{\sktsub {(\Inst
      {\Ivar'} \settv \skel)}{\typ_2}} \term {\envvar}{\typ_2}{(\Instc {\Ivar'}
    \settv \typ \cons \consinter ((\Inst {\Ivar'} \settv \typ) \consleq
    \typ_2))}$, i.e., $\jtyp{\Inst {\Ivar} \settv \skel} \term {\envvar}{\Inst
    {\Ivar} \settv \typ}{\Instc {\Ivar} \settv \typ \cons}$, as required.

  If $\Ivar = \iapp \subv {\settv'}{\Ivar'}$, then $\jtyp{\Inst {\Ivar'} \settv
    \skel} \term {\envvar}{\Inst {\Ivar'} \settv \typ}{\Instc {\Ivar'} \settv
    \typ \cons}$ holds by induction. Because $\ftv \envvar \subseteq \settv$, we
  have $\ftv \envvar \subseteq \settv \cup \settv'$, so by rule \TRsintro, we
  obtain
  $$\jtyp{\skels \subv
    {\settv \cup \settv'}{(\Inst {\Ivar'} \settv \skel)}} \term {\envvar}{\styp \subv
    {\settv \cup \settv'}{(\Inst {\Ivar'} \settv \typ)}}{\conss \subv
    {\settv \cup \settv'}{\Inst {\Ivar'} \settv \typ}{(\Instc {\Ivar'}
    \settv \typ \cons)}},$$ i.e., $\jtyp{\Inst {\Ivar} \settv \skel} \term {\envvar}{\Inst
    {\Ivar} \settv \typ}{\Instc {\Ivar} \settv \typ \cons}$, as required.

  If $\Ivar = \iquant \tvar {\Ivar'}$, then by induction $\jtyp{\Inst {\Ivar'}
    \settv \skel} \term {\envvar}{\Inst {\Ivar'} \settv \typ}{\Instc {\Ivar'}
    \settv \typ \cons}$ holds. If $\tvar \in \settv$, then we have the required
  result. If $\tvar \notin \settv$, then because $\ftv \envvar \subseteq
  \settv$, we have $\tvar \notin \ftv \envvar$. Hence, we have $\jtyp{\skelquant
    \tvar {\Inst {\Ivar'} \settv \skel}} \term {\envvar}{\uquant \tvar{\Inst
      {\Ivar'} \settv \typ}}{\consquant \tvar {\Instc {\Ivar'} \settv \typ
      \cons}}$, by rule \TRforall, i.e., $\jtyp{\Inst {\Ivar} \settv \skel}
  \term {\envvar}{\Inst {\Ivar} \settv \typ}{\Instc {\Ivar} \settv \typ \cons}$,
  as required.

\end{proof}

\begin{theorem}
  \label{l:exp-preserves-typings}
  If $\jtyp \skel \term \envvar \typ \cons$ then $\jtyp{\subst \svar \skel}
  \term {\subst \svar \envvar}{\subst \svar \typ}{\subst \svar \cons}$.
\end{theorem}

\begin{proof}
  We proceed by induction on $\skel$.
  
  Suppose $\skel = \skelvar \expvar \envvar$; we have $\jtyp{\skelvar \expvar
    {\subst \svar \envvar}} \expvar {\subst \svar \envvar}{\apply {(\subst \svar
      \envvar)} \expvar} \consomega$ by rule \TRvar, and $\apply {(\subst \svar
      \envvar)} \expvar = \subst \svar {(\apply \envvar \expvar)}$hence the result
  holds.

  If $\skel = \skelabs \expvar {\skel'}$, then we have $\jtyp{\skelabs \expvar
    {\skel'}}{\lamb \expvar \term} \envvar {\fun{\typ_1}{\typ_2}} \cons$ with
  $\jtyp{\skel'}{\term} {\eextend \envvar {\eassign \expvar {\typ_1}}} {\typ_2}
  \cons$. We have $\jtyp{\subst \svar {\skel'}}{\term} {\eextend {\subst \svar
      \envvar} {\eassign \expvar {\subst \svar {\typ_1}}}} {\subst \svar
    {\typ_2}}{\subst \svar \cons}$ by induction, consequently we have
  $\jtyp{\skelabs \expvar {\subst \svar {\skel'}}}{\lamb \expvar \term}{\subst
    \svar \envvar}{\fun{\subst \svar {\typ_1}}{\subst \svar {\typ_2}}}{\subst
    \svar \cons}$ by rule \TRabs, i.e. $\jtyp{\subst \svar {\skelabs \expvar
      {\skel'}}}{\lamb \expvar \term}{\subst \svar \envvar}{\subst \svar
    {(\fun{\typ_1}{\typ_2})}}{\subst \svar \cons}$, as required.

  \begin{sloppypar}
  If $\skel = \skel_1 \skelapp \skel_2$, then we have $\jtyp {\skel_1 \skelapp
    \skel_2}{\term_1 \app \term_2}{\envvar}{\typ_2}{(\cons_1 \consinter
    \cons_2)}$ with $\jtyp
  {\skel_1}{\term_1}{\envvar}{\fun{\typ_1}{\typ_2}}{\cons_1}$ and $\jtyp
  {\skel_2}{\term_2}{\envvar}{\typ_1}{\cons_2}$. By induction, we have $\jtyp
  {\subst \svar {\skel_1}}{\term_1}{\subst \svar \envvar}{\subst \svar
    {(\fun{\typ_1}{\typ_2}})}{\subst \svar {\cons_1}}$ and $\jtyp {\subst \svar
    {\skel_2}}{\term_2}{\subst \svar \envvar}{\subst \svar {\typ_1}}{\subst
    \svar {\cons_2}}$. Since we have $\subst \svar {(\fun{\typ_1}{\typ_2})} =
  \fun{\subst \svar {\typ_1}}{\subst \svar {\typ_2}}$, we have $\jtyp {\subst
    \svar {\skel_1} \skelapp \subst \svar {\skel_2}}{\term_1 \app
    \term_2}{\subst \svar {\envvar}}{\subst \svar {\typ_2}}{\subst \svar
    {\cons_1} \consinter \subst \svar {\cons_2}}$ by rule \TRapp, i.e.  $\jtyp
  {\subst \svar {(\skel_1 \skelapp \skel_2)}}{\term_1 \app \term_2}{\subst \svar
    \envvar}{\subst \svar {\typ_2}}{\subst \svar {(\cons_1 \consinter
      \cons_2)}}$, as required.
  \end{sloppypar}

  If $\skel = \skelquant \tvar {\skel'}$, then we have $\jtyp {\skelquant
    \tvar {\skel'}} \term {\envvar}{\uquant \tvar \typ}{\cons}$ with
  $\jtyp {\skel'} \term \envvar \typ \cons$ and $\tvar \notin \ftv \envvar$. By
  $\alpha$-conversion, we can assume that $\tvar \notin \ftv\svar$.  By
  induction, we have $\jtyp {\subst \svar {\skel'}} \term {\subst \svar
    \envvar}{\subst \svar \typ}{\subst \svar \cons}$. Since we have $\tvar
  \notin \ftv \envvar$ and $\tvar \notin \ftv \svar$, we have $\tvar \notin \ftv
  {\subst \svar \envvar}$.  By rule \TRforall, we have $\jtyp{\skelquant
    {\tvar} {\subst \svar {\skel'}}} \term {\subst \svar \envvar}{\uquant
    \tvar {\subst \svar \typ}}{\subst \svar \cons}$, i.e. $\jtyp{\subst
    \svar {\skelquant \tvar {\skel'}}} \term {\subst \svar \envvar}{\subst
    \svar {\uquant \tvar \typ}}{\subst \svar \cons}$, hence the result
  holds.

  If $\skel = \sktsub {{\skel'}}{\typ_2}$, then we have $\jtyp {\skel'} \term
  \envvar {\typ_1}{\cons'}$ with $\cons = \cons' \consinter \typ_1 \consleq
  \typ_2$. By induction, we have $\jtyp {\subst \svar {\skel'}} \term {\subst
    \svar \envvar}{\subst \svar {\typ_1}}{\subst \svar {\cons'}}$.  Applying rule
  \TRsub, we obtain $\jtyp {\sktsub{\subst \svar {{\skel'}}}{\subst \svar
      {\typ_2}}} \term {\subst \svar \envvar}{\subst \svar {\typ_2}}{(\subst
    \svar {\cons'} \consinter \subst \svar {\typ_1} \consleq \subst \svar
    {\typ_2})}$, i.e. $\jtyp {\subst \svar {\sktsub{{\skel'}}{\typ_2}}} \term
  {\subst \svar \envvar}{\subst \svar {\typ_2}}{\subst \svar {(\cons' \consinter
      \typ_1 \consleq \typ_2)}}$, as required.

  If $\skel = \skels \subv \settv {\skel'}$, then we have $\jtyp {\skels \subv
    \settv {\skel'}} \term \envvar {\styp \subv \settv \typ}{\conss \subv
    \settv \typ {\cons'}}$ with $\jtyp {\skel'} \term \envvar \typ
  {\cons'}$ and $\ftv \envvar \subseteq \settv$. By induction, the judgement $\jtyp
  {\subst \svar {\skel'}} \term {\subst \svar \envvar}{\subst \svar \typ}{\subst
    \svar {\cons'}}$ holds. Because $\ftv \envvar \subseteq \settv$, we have $\ftv
  {\subst \svar \envvar} \subseteq \ftv{\subst \svar \settv}$, therefore by
  Lemma \ref{l:iinst-preserves-typings} we have $$\jtyp {\Inst {\subst \svar
      \subv}{\ftv{\subst \svar \settv}}{\subst \svar {\skel'}}} \term {\subst
    \svar \envvar}{\Inst {\subst \svar \subv}{\ftv {\subst \svar \settv}}{\subst
      \svar \typ}}{\Instc {\subst \svar
      \subv}{\ftv {\subst \svar \settv}}{\subst \svar \typ}{\subst \svar \cons}},$$ hence we
  have $\jtyp {\subst \svar {(\skels \subv \settv {\skel'})}} \term {\subst
    \svar \envvar}{\subst \svar {(\styp \subv \settv \typ)}}{\subst \svar
    \cons}$, as required.
  
\end{proof}


\section{Subject reduction}
\label{a:subject-reduction}

We prove subject reduction for the System F subtyping. We define an equivalent
type system where we turn equalities on types into explicit subtyping rules; we
then prove subject reduction in the equivalent type system.

\subsection{An equivalent type system \systemfsneq}

Typing judgements of \systemfsneq, written $\jneq \skel \term \envvar \typ$, are
derived according to rules given in Figure \ref{f:TR-fsneq}. The typing
rules are the same as in the original type system, except that we add an
environment subtyping rule (restricted to equality subtyping), and we do not
mention constraints anymore: in subtyping rules, we consider only solved
constraints, and the subtyping proofs are specified by subtyping skeletons
$\sksub$. Subtyping skeletons are mentionned in the skeletons $\skneqtsub \skel
\sksub$ and $\skneqenvsub \skel {\eassign y \sksub}$.

\begin{figure}[t]
\begin{mathpar}
  \inferrule
  {}
  {\jneq {\skelvar \expvar \envvar} \expvar {\envvar}
    {\apply \envvar \expvar}}~~\TRvar
  \and
  \inferrule
  {\jneq \skel \term {\eextend \envvar {\eassign \expvar {\typ_1}}} {\typ_2}}
  {\jneq {\skelabs \expvar \skel}{\lamb \expvar \term} \envvar
    {\fun{\typ_1}{\typ_2}}}~~\TRabs
  \and
  \inferrule
  {\jneq {\skel_1}{\term_1}{\envvar}{\fun{\typ_1}{\typ_2}} \\ 
    \jneq {\skel_2}{\term_2}{\envvar}{\typ_1} }
  {\jneq {\skel_1 \skelapp \skel_2}{\term_1 \app \term_2}{\envvar}{\typ_2}}~~\TRapp
  \and
  \inferrule
  {\jneq \skel \term \envvar \typ \\ \tvar \notin \ftv \envvar}
  {\jneq {\skelquant \tvar \skel} \term {\envvar}{\uquant \tvar 
      \typ}}~~\TRforall
  \and
  \inferrule
  {\jneq \skel \term \envvar {\typ_1} \\ \jsub \sksub {\typ_1} \vleq {\typ_2} }
  {\jneq {\skneqtsub \skel \sksub} \term {\envvar}{\typ_2}}~~\TRsolvedsub
  \and
  \inferrule
  {\jneq \skel \term {\eextend \envvar {\eassign y {\typ_1}}} \typ \\ \jsub
    \sksub {\typ_2} \leqeq {\typ_1}}
  {\jneq {\skneqenvsub \skel {\eassign y \sksub}} \term {\eextend \envvar
      {\eassign y {\typ_2}}} \typ}~~\TRenvsub
  \and
  \inferrule
  {\jneq \skel \term \envvar \typ \\ \ftv \envvar \subseteq \settv}
  {\jneq {\skels \subv \settv \skel} \term \envvar{\styp \subv \settv \typ}}~~\TRsintro
\end{mathpar}
\caption{Typing rules of \systemfsneq}
\label{f:TR-fsneq}
\end{figure}

The subtyping rules and the corresponding subtyping skeletons are given in
Fig. \ref{f:SR-fsneq}; we define two subtyping judgements $\jsub \sksub {\typ_1}
\leq {\typ_2}$ (for the regular subtyping relation, from the original type
system) and $\jsub \sksub {\typ_1} \leqeq {\typ_2}$ (which deals with equalities
on types). We let $\vleq$ ranges over the two subtyping relations. 

\begin{figure}[t]
\textbf{Regular subtyping}
\begin{mathpar}
  \inferrule
  {}
  {\jsub{\sksinst {\uquant \tvar {\typ_1}}{\typ_2}}\uquant \tvar
    {\typ_1} \leq \subst {\assign \tvar {\typ_2} \sepsubs \idsubs}{\typ_1}}
\end{mathpar}
\textbf{Equality subtyping}
\begin{mathpar}
  \inferrule
  {}
  {\jsub{\sksquantcomm{\typ}}{\uquant
      {\tvar_1}{\uquant{\tvar_2} \typ}}
    \leqeq {\uquant {\tvar_2}{\uquant {\tvar_1} \typ}}}
  \and
  \inferrule
  {\tvar \notin \ftv \typ}
  {\jsub{\sksdummyin{\typ}}{\typ
    \leqeq {\uquant \tvar \typ}}}
  \and
  \inferrule
  {\tvar \notin \ftv \typ}
  {\jsub{\sksdummyelim{\typ}}{\uquant
      \tvar \typ \leqeq \typ}}
  \and
  \inferrule
  {\jsub {\sksub_1}{\typ_2} \leqeq {\typ_1} \\ \jsub {\sksub_2}{\typ_3} \leqeq {\typ_4}}
  {\jsub {\sksfun{\sksub_1}{\sksub_2}}{\fun{\typ_1}{\typ_3}} \leqeq
    {\fun{\typ_2}{\typ_4}}}
  \and
  \inferrule
  {\jsub {\sksub}{\typ_1} \leqeq {\typ_2}}
  {\jsub {\skssubv \subv \settv \sksub}{\styp \subv \settv {\typ_1}} \leqeq {\styp \subv
     \settv  {\typ_2}}} 
  \and
  \inferrule
  {\jsub {\sksub}{\typ_1} \vleq {\typ_2}}
  {\jsub {\sksquant \tvar \sksub}{\uquant \tvar {\typ_1}} \leqeq {\uquant \tvar {\typ_2}}}
\end{mathpar}
\caption{Subtyping rules of \systemfsneq}
\label{f:SR-fsneq}
\end{figure}

\begin{lemma}
  If $\jsub \sksub {\typ_1} \leqeq {\typ_2}$ then $\jsub {\sksub'}{\typ_2} \leqeq {\typ_1}$.
\end{lemma}

\begin{proof}
  By induction on the derivation of $\jsub \sksub {\typ_1} \leqeq {\typ_2}$

\end{proof}

\begin{lemma}
  We have $\typ_1 \leq \typ_2$ in \systemfs iff there exists $\sksub$ such that $\jsub \sksub
  {\typ_1} \vleq {\typ_2}$ in \systemfsneq. 
\end{lemma}

\begin{proof}
  By induction on the derivation of $\typ_1 \leq \typ_2$ and by induction on
  the derivation of $\jsub \sksub {\typ_1} \vleq {\typ_2}$

\end{proof}

\begin{lemma}
  \label{l:corr-type-systems}
  We have $\jneq {\skel} \term \envvar \typ$ iff $\jtyp {\skel'} \term
  {\envvar}{\typ} \cons$ where $\cons$ is solved.
\end{lemma}

\begin{proof}
  By induction on the derivation of $\jneq \skel \term \envvar \typ$, and by
  induction on the derivation of $\jtyp {\skel'} \term \envvar \typ \cons$.

\end{proof}

We now prove subject reduction for \systemfsneq.

\subsection{Skeleton transformation}

In Figure \ref{f:def-T}, we define a transformation $\trsymb$ on skeletons which
preserves typings while removing the unnecessary uses of the result subtyping
rule \TRsolvedsub. We write $\envvar_1 \leqeq \envvar_2$ iff for all $\expvar$,
we have $\apply{\envvar_1}{\expvar} \leqeq^+ \apply{\envvar_2}{\expvar}$. For
all $\skel$, $\envvar_1$, $\envvar_2$ such that $\envvar_2 = \tenv{\skel}$ and
$\envvar_1 \leqeq \envvar_2$, there exists a skeleton, written $\skneqenvsub
\skel {\envvar_1 \leqeq \envvar_2}$, obtained from $\skel$ by repeated use of
the environment subtyping rule and such that $\tenv{\skneqenvsub \skel
  {\envvar_1 \leqeq \envvar_2}}=\envvar_1$. We informally use this notation in
proof when it is more convenient than the original one.

\begin{figure}[t]
\textbf{Constructor rules}
\begin{mathpar}
  \inferrule
  {}
  {\tr{\lamb \expvar \skel}=\lamb \expvar \skel}
  \and
  \inferrule
  {}
  {\tr{\skelquant \tvar \skel}=\skelquant \tvar {\tr \skel}}
  \and
  \inferrule
  {\tr{\skel}=\skelabs \expvar {\skel'} }
  {\tr{\skneqenvsub \skel {\eassign y \sksub}}=\skelabs \expvar
    {(\skneqenvsub {\skel'} {\eassign y \sksub} )}}
  \and
  \inferrule
  {\tr\skel = \skels \subv \settv {\skel'}}
  {\tr {\skneqenvsub \skel {\eassign y \sksub}}=\skels \subv \settv {(\skneqenvsub
      {\skel'}{\eassign y \sksub})}}
  \and
  \inferrule
  {\tr\skel = \skelquant \tvar {\skel'}}
  {\tr {\skneqenvsub \skel {\eassign y \sksub}}=\skelquant \tvar {(\skneqenvsub
      {\skel'}{\eassign y \sksub})}}
  \and
  \tr {\skels \subv \settv \skel}= \skels \subv \settv \skel
\end{mathpar}
\textbf{Original subtyping rules}
\begin{mathpar}
  \inferrule
  {\tr \skel = \skelquant \tvar {\skel'}
  }
  {\tr{\skneqtsub \skel {\sksinst {\uquant \tvar {\typ_1}}{\typ_2}} }= \tr
    {
      {\subst {\assign \tvar {\typ_2}}{\skel'}}}} 
\end{mathpar}
\textbf{Equality subtyping rules}
\begin{mathpar}
  \inferrule
  {\tr {\skel} = \skels \subv \settv {\skel'}}
  {\tr {\skneqtsub \skel {\skssubv \subv \settv \sksub}}=\skels \subv \settv {(\skneqtsub
      {(\skel')}\sksub)}}
  \and
  \inferrule
  {\tr \skel = \skelquant {\tvar_1}{\skel_1} \\ \tr{\skel_1} =
    \skelquant {\tvar_2}{\skel_2}} 
  {\tr {\skneqtsub \skel {\sksquantcomm{\typ}}}= 
    \skelquant {\tvar_2}{\skelquant{\tvar_1}{\skel_2}} }
  \and
  \inferrule
  {\tr{\skel}=\skelabs \expvar \skel}
  {\tr{\skneqtsub \skel {\sksfun{\sksub_1}{\sksub_2}}}=\skelabs \expvar
    {(\skneqenvsub {(\skneqtsub \skel {\sksub_2})}{\eassign \expvar
        {\sksub_1}})} }
  \and
  \inferrule
  {\tr \skel = \skelquant \tvar {\skel_1} \\
    \tr{\skneqtsub {\skel_1}{\sksub}}=\skel'_1
  }
  {\tr{\skneqtsub \skel {\sksquant \tvar {\sksub}} }= \skelquant \tvar {\skel'_1}} 
  \and
  \inferrule
  {\tvar \notin \ftv \typ}
  {\tr {\skneqtsub \skel {\sksdummyin \sksub}} = \skelquant \tvar \skel}
  \and
  \inferrule
  {\tr \skel = \skelquant \tvar {\skel'}}
  {\tr {\skneqtsub \skel {\sksdummyelim{\typ}}}= \tr{\skel'}}
\end{mathpar}
\caption{Definition of $\trsymb$}
\label{f:def-T}
\end{figure}

\subsection{Induction principle on skeletons}

\begin{figure}[t]
\begin{mathpar}
  \sz{\skelabs \expvar \skel} = 1 
  \and
  \sz{\skelquant \tvar \skel}=1+ \sz\skel
  \and
  \sz{\skneqenvsub \skel {\eassign y \sksub}}=1+ \sz\skel
  \and
  \sz{\skels \subv \settv \skel} = 1+\sz \skel
  \and
  \sz{\skneqtsub \skel {\sksinst{\uquant \tvar {\typ_1}}{\typ_2}}
  }=1+\sz{\skel}
  \and
  \sz{\skneqtsub \skel {\sksquant{\tvar} {\sksub}}
  }=1+\sz{\skneqtsub \skel {\sksub}}
  \and
  \sz{\skneqtsub \skel {\sksquantcomm{\typ}}}=1+\sz\skel
  \and
  \sz{\skneqtsub \skel {\sksdummyin{\typ}}}=2+\sz\skel
  \and
  \sz{\skneqtsub \skel {\sksdummyelim{\typ}}}=1+\sz\skel
  \and
  \sz{\skneqtsub \skel {\skssubv \subv \settv \sksub}}=2+\sz{\skneqtsub \skel \sksub}
  \and
  \sz{\skneqtsub \skel {\sksfun{\sksub_1}{\sksub_2}} }=1+\sz\skel
\end{mathpar}
\caption{Definition of \textsf{sz}}
\label{f:def-sz}
\end{figure}

Proofs on skeletons are by induction on the size $\sz \skel$, defined in Figure
\ref{f:def-sz}. We prove that transformation $\trsymb$ makes the size
decrease. We need first some preliminary results.

\begin{lemma}
\label{l:sz-prop}

  \begin{itemize}
  \item We have $\sz\skel \geq 1$.
  \item $\sz{\skel}=1$ iff $\skel = \skelabs \expvar {\skel'}$.
  \item For all $\skel$, we have $\sz{\subst{\assign \tvar {\typ} \sepsubs
        \idsubs}{\skel}}=\sz{\skel}$.
  \item For all $\sksub$, we have $\sz{\skneqtsub \skel \sksub} >\sz\skel$.
  \item For all $\sksub$, if $\sz{\skel_1} \leq \sz{\skel_2}$ then
    $\sz{\skneqtsub {\skel_1} \sksub} \leq \sz{\skneqtsub {\skel_2} \sksub}$ 
  \end{itemize}
\end{lemma}

\begin{proof}
  The first three items are easy. The fourth item is by induction on $\sksub$. 

  Suppose $\sksub = \sksquant \tvar {\sksub'}$. By induction we have $\sz
  {\skneqtsub {\skel}{\sksub'}} > \sz \skel$, consequently we have $\sz{\skneqtsub \skel
    \sksub} > \sz \skel+1 > \sz\skel$. The remaining cases are easy.\\

  The last item is by induction on $\sksub$.
  
  Suppose $\sksub = \sksquant \tvar {\sksub'}$. Let $\skel_1, \skel_2$ such that
  $\sz{\skel_1} \leq \sz{\skel_2}$. By induction we have $\sz {\skneqtsub
    {\skel_1}{\sksub'}} \leq \sz {\skneqtsub {\skel_2}{\sksub'}}$.  Consequently
  we have $\sz {\skneqtsub {\skel_1}{\sksub'}}+1 \leq \sz {\skneqtsub
    {\skel_2}{\sksub'}}+1$, i.e. $\sz {\skneqtsub {\skel_1} \sksub} \leq \sz
  {\skneqtsub {\skel_2} \sksub}$, as wished. In the remaining cases, the size of
  $\sz{\skneqtsub {\skel_1} \sksub}$ differs from $\sz{\skel_1}$ by a positive
  integer $C(\sksub)$, i.e. we have $\sz{\skneqtsub {\skel_1} \sksub}=\sz
  {\skel_1}+C(\sksub) \leq \sz {\skel_2}+C(\sksub) = \sz{\skneqtsub {\skel_2}
    \sksub}$, hence the result holds.
  
\end{proof}

\begin{lemma}
  \label{l:sz-decr}
  We have $\sz{\tr\skel} \leq \sz{\skel}$.
\end{lemma}

\begin{proof}
  We proceed by induction on $\sz{\skel}$. If $\sz{\skel}=1$, then $\skel =
  \lamb \expvar \skel$, so $\tr\skel = \skel$, hence the result holds. Suppose
  that the result holds for $\sz\skel
  \leq n$; we prove it for $\sz\skel = n+1$ by case analysis on $\skel$.\\

  If $\skel=\skelquant \tvar {\skel'}$, then $\tr\skel = \skelquant \tvar
  \tr{\skel'}$. We have $\sz{\skel'} < \sz{\skel}$, so by induction we have
  $\sz{\tr{\skel'}} \leq \sz{\skel'}$. Consequently we have
  $\sz{\tr\skel}=\sz{\tr{\skel'}}+1 \leq \sz{\skel'}+1 = \sz{\skel}$, hence the
  result holds.\\

  If $\skel=\skels \subv \settv {\skel'}$, then $\tr \skel = \skel$, hence the result
  holds.\\

  If $\skel=\skneqenvsub {\skel'}{\eassign y \sksub}$, then we distinguish
  several cases. If $\tr{\skel'}=\lamb \expvar {\skel''}$ for some $\skel''$,
  then we have $\tr{\skel}=\lamb \expvar {(\skneqenvsub {\skel''}{\eassign y
      \sksub})}$; we have $\sz{\tr\skel}=1 \leq \sz{\skel}$ as required. If
  $\tr{\skel'}=\skels \subv \settv {\skel''}$, then $\tr\skel = \skels \subv
  \settv {(\skneqenvsub {\skel''}{\eassign y \sksub})}$. By induction we have
  $\sz{\tr{\skel'}} \leq \sz{\skel'}$, i.e. $\sz{\skel''}+1 \leq
  \sz{\skel'}$. Therefore we have $\sz{\tr \skel}= 2+\sz{\skel''} \leq
  \sz{\skel'}+1 \leq \sz{\skel}$, as required. If $\tr{\skel'}=\skelquant \tvar
  {\skel''}$, then the proof is similar to
  the previous case. \\
 
  If $\skel = \skneqtsub {\skel'}{\sksinst {\uquant \tvar
      {\typ_1}}{\typ_2}}$, then $\tr{\skel}= \tr {\subst {\assign \tvar
      {\typ_2}}{\skel''}}$ with $\tr {\skel'} = \skelquant \tvar
  {\skel''}$. We have $\sz{\skel'}<\sz\skel$, so by induction we have
  $\sz{\tr{\skel'}} \leq \sz{\skel'}$, i.e. $\sz{\skel''}+1 \leq
  \sz{\skel'}$. By Lemma \ref{l:sz-prop}, we have $\sz{\subst{\assign \tvar
      {\typ_2}}{\skel''}}=\sz{\skel''}$, therefore we have $\sz{\subst {\assign
      \tvar {\typ_2}}{\skel''}}=\sz{\skel''} \leq \sz{\skel'} < \sz{\skel}$, so
  by induction we have $\sz{\tr{\subst {\assign \tvar {\typ_2}}{\skel''}}} \leq
  \sz{\subst {\assign \tvar {\typ_2}}{\skel''}} \leq \sz\skel$, i.e. $\sz{\tr
    \skel}\leq \sz\skel$, as required.\\

  If $\skel = \skneqtsub{\skel'}{\sksquant \tvar {\sksub}}$, then
  $\tr{\skel}= \skelquant \tvar {\skel'_1}$ with $\tr{\skel'} = \skelquant
  \tvar {\skel_1}$, and $\tr{\skneqtsub {\skel_1}{\sksub}}=\skel'_1$. We
  have $\sz{\skel'} <\sz{\skel}$, so by induction, we have $\sz{\tr{\skel'}}
  \leq \sz{\skel'}$, i.e.  $\sz{\skel_1}+1 \leq \sz{\skel'}$. Consequently we
  have $\sz{\skel_1}\leq \sz{\skel'}$, so by Lemma \ref{l:sz-prop} we have
  $\sz{\skneqtsub {\skel_1}{\sksub}}\leq \sz{\skneqtsub{\skel'}{\sksub}}$. By
  the definition of $\sz{}$, we have then $\sz{\skneqtsub {\skel_1}{\sksub_1}} <
  \sz{\skel}$, so by induction we have $\sz{\tr{\skneqtsub {\skel_1}{\sksub}}}
  \leq \sz{\skneqtsub {\skel_1}{\sksub}}$. Finally we have $\sz{\tr
    \skel}=\sz{\skel'_1}+1 = \sz{\tr{\skneqtsub{\skel_1}{\sksub}}}+1 \leq
  \sz{\skneqtsub {\skel_1}{\sksub}}+1 \leq \sz{\skneqtsub {\skel'}{\sksub}}+1
  \leq \sz{\skel}$, hence the result holds.\\

  If $\skel = \skneqtsub {{\skel'}}{\skssubv \subv \settv \sksub}$, then $\tr\skel =
  \skels \subv \settv {(\skneqtsub {{\skel'}}{\sksub})}$. We have $\sz{\tr \skel}=
  \sz{\skneqtsub {{\skel'}}{\sksub}}+1 \leq \sz \skel$, hence the result holds.\\

  If $\skel = \skneqtsub {\skel'}{\sksquantcomm{\typ}}$, then $\tr \skel =
  \skelquant {\tvar_2}{\skelquant {\tvar_1}{\skel_2}}$ with $\tr {\skel'} =
  \skelquant {\tvar_1}{\skel_1}$ and $\tr{\skel_1} = \skelquant
  {\tvar_2}{\skel_2}$. By induction we have $\sz {\tr {\skel'}} \leq \sz
  {\skel'}$, i.e. $\sz {\skel_1}+1 \leq \sz{\skel'}$. Applying the induction
  hypothesis on $\skel_1$, we have $\sz{\tr{\skel_1}} \leq \sz{\skel_1}$,
  i.e. $\sz{\skel_2}+1 \leq \sz{\skel_1}$. Consequently we have $\sz {\tr
    \skel}= \sz{\skel_2}+2 \leq
  \sz{\skel_1}+1 \leq \sz{\skel'} \leq \sz{\skel}$, hence the result holds.\\

  If $\skel = \skneqtsub {\skel'}{\sksdummyin{\typ}}$, then $\tr \skel =
  \skelquant \tvar {\skel'}$, where $\tvar \notin \ftv{\skel'}$. We have $\sz {\tr
    \skel}= \sz{\skel'}+1 \leq \sz{\skel'}+2 = \sz{\skel}$, hence the result
  holds.\\

  If $\skel = \skneqtsub {\skel'}{\sksdummyelim{\typ}}$, then $\tr \skel = \tr
  {\skel''}$ with $\tr{\skel'} = \skelquant \tvar {\skel''}$. By induction we
  have $\sz {\tr {\skel'}} \leq \sz {\skel'}$, i.e. $\sz {\skel''}+1 \leq
  \sz{\skel'}$. Therefore we have $\sz {\skel''} < \sz \skel$, so by induction
  we get $\sz{\tr{\skel''}} \leq \sz{\skel''}$. We have $\sz {\tr \skel}=
  \sz{\tr {\skel''}} \leq \sz{\skel''} \leq \sz{\skel'} \leq \sz \skel$, hence
  the result holds.\\

  If $\skel = \skneqtsub {\skel'}{\sksfun{\sksub_1}{\sksub_2}}$, then there
  exists $\skel''$ such that $\tr\skel = \lamb \expvar {\skel''}$. We have
  $\sz{\tr\skel}=1 \leq \sz\skel$ as required.

\end{proof}

\begin{lemma}
  \label{l:tr-abs}
  For all $\skel$, 
  \begin{itemize}
  \item If $\jneq \skel {\lamb \expvar \term} \envvar {\fun {\typ_1}{\typ_2}}$
    then there exists $\skel'$ such that $\tr\skel=\lamb \expvar {\skel'}$ and $\jneq
    {\tr\skel}{\lamb \expvar \term} \envvar {\fun {\typ_1}{\typ_2}}$.
  \item If $\jneq \skel {\lamb \expvar \term} \envvar {\uquant \tvar \typ}$ then
    there exists $\skel'$ such that $\tr\skel= \skelquant \tvar {\skel'}$ and
    $\jneq {\tr\skel}{\lamb \expvar \term} \envvar {\uquant \tvar \typ}$.
  \item If $\jneq \skel {\lamb \expvar \term} \envvar {\styp \subv \settv \typ}$ then
    there exists $\skel'$ such that $\tr\skel= \skels \subv \settv {\skel'}$ and $\jneq
    {\tr\skel}{\lamb \expvar \term} \envvar {\styp \subv \settv \typ}$.
  \end{itemize}
\end{lemma}

\begin{proof}
  We proceed by induction on $\sz{\skel}$. If $\sz{\skel}=1$, then $\skel =
  \lamb \expvar \skel$, so $\tr\skel =
  \skel$, and the first item of the lemma hold. Suppose that the
  result holds for $\sz\skel \leq n$; we prove it for $\sz\skel = n+1$ by case
  analysis on $\skel$.\\
  
  If $\skel=\skelquant \tvar {\skel'}$, then by the type system we have
  $\jneq \skel {\lamb \expvar \term} {\envvar}{\uquant \tvar {\typ'}}$
  with $\jneq {\skel'}{\lamb \expvar \term} {\envvar}{\typ'}$. We have
  $\sz{\skel'} < \sz{\skel}$, so by induction we have $\jneq {\tr{\skel'}}{\lamb
    \expvar \term} {\envvar}{\typ'}$. By rule \TRforall we have $\jneq
  {\skelquant \tvar \tr{\skel'}}{\lamb \expvar \term} {\envvar}{\uquant
    \tvar {\typ'}}$, i.e. $\jneq {\tr{\skel}}{\lamb \expvar
    \term} {\envvar}{\uquant \tvar {\typ'}}$, as required.\\

  If $\skel=\skneqenvsub {\skel'}{\envvar \leqeq \envvar'}$, then we distinguish
  several cases. If $\jneq \skel {\lamb \expvar \term} \envvar {\fun
    {\typ_1}{\typ_2}}$, then by the type system we have $\jneq {\skel'}{\lamb
    \expvar \term}{\envvar'}{\fun {\typ_1}{\typ_2}}$. We have $\sz{\skel'} <
  \sz\skel$, so by induction there exists $\skel''$ such that $\tr{\skel'}=\lamb
  \expvar {\skel''}$ and $\jneq {\lamb \expvar {\skel''}}{\lamb \expvar
    \term}{\envvar'}{\fun {\typ_1}{\typ_2}}$. By rule \TRabs, we have $\jneq
  {\skel''}{\term}{\eextend {\envvar'}{\eassign \expvar {\typ_1}}}{\typ_2}$, so
  we have $\jneq {\skneqenvsub \skel {\envvar \leqeq \envvar'}}{\term}{\eextend
    \envvar {\eassign \expvar {\typ_1}}}{\typ_2}$ by environment subtyping,
  hence we have $\jneq {\lamb \expvar {\skneqenvsub {\skel'} {\envvar \leqeq
        \envvar'}}}{\lamb \expvar \term}{\envvar}{\fun {\typ_1}{\typ_2}}$. By
  definition of $\trsymb$, we have $\tr\skel = \lamb \expvar {\skneqenvsub \skel
    {\envvar \leqeq \envvar'}}$, hence the result holds.

  If $\jneq \skel {\lamb \expvar \term} \envvar {\styp \subv \settv {\typ'}}$,
  then by environment subtyping we have $\jneq {\skel'}{\lamb \expvar
    \term}{\envvar'}{\styp \subv \settv {\typ'}}$. By induction we have
  $\tr{\skel'}=\skels \subv \settv {\skel''}$ and $\jneq {\tr{\skel'}}{\lamb
    \expvar \term}{\envvar'}{\styp \subv \settv {\typ'}}$. By rule \TRsintro, we
  have $\jneq {\skel''}{\lamb \expvar \term}{\envvar'}{\typ'}$; therefore we
  have $\jneq {\skels \subv \settv {(\skneqenvsub {\skel''}{\envvar \leqeq
        \envvar'})}}{\lamb \expvar \term}{\envvar}{\styp \subv \settv {\typ'}}$,
  i.e. $\jneq {\tr\skel}{\lamb \expvar \term} \envvar {\styp \subv \settv
    {\typ'}}$, as required. The proof is similar in the case $\typ =
  \uquant \tvar {\typ'}$.\\

  If $\skel = \skels \subv \settv {\skel'}$, then we have $\jneq \skel {\lamb
    \expvar \term} \envvar {\styp \subv \settv {\typ'}}$, and since $\tr\skel =
  \skel$, the result holds.\\

  If $\skel = \skneqtsub {\skel'}{\sksinst {\uquant \tvar {\typ_1}}{\typ_2}}$,
  then by rule \TRsolvedsub we have $\jneq{\skel}{\lamb \expvar
    \term}{\envvar}{\subst{\assign \tvar {\typ_2}}{\typ_1}}$ and
  $\jneq{\skel'}{\lamb \expvar \term}{\envvar}{\uquant \tvar {\typ_1}}$. We have
  $\sz{\skel'} < \sz{\skel}$, so by induction there exists $\skel''$ such that
  $\tr{\skel'}=\skelquant \tvar {\skel''}$ and $\jneq {\tr{\skel'}}{\lamb
    \expvar \term}{\envvar}{\uquant \tvar {\typ_1}}$. By rule \TRforall, we have
  $\jneq {\skel''}{\lamb \expvar \term}{\envvar}{\typ_1}$ and $\tvar \notin \ftv
  \envvar$, so by Lemma \ref{l:exp-preserves-typings}, we have $\jneq
  {\subst{\assign \tvar {\typ_2}}{\skel''}}{\lamb \expvar
    \term}{\envvar}{\subst{\assign \tvar {\typ_2}}{\typ_1}}$. By Lemma
  \ref{l:sz-prop}, we have $\sz{\subst{\assign \tvar
      {\typ_2}}{\skel''}}=\sz{\skel''} \leq \sz{\skel'} < \sz\skel$, so by
  induction we have $\jneq {\tr{\subst{\assign \tvar {\typ_2}}{\skel''}}}{\lamb
    \expvar \term}{\envvar}{\subst{\assign \tvar {\typ_2}}{\typ_1}}$, and the
  shape of $\tr{\subst{\assign \tvar {\typ_2}}{\skel''}}$ follows the shape of
  $\subst{\assign \tvar {\typ_2}}{\typ_1}$. By definition of $\trsymb$, we have
  $\tr\skel=\tr{\subst{\assign \tvar {\typ_2}}{\skel''}}$, hence the result holds.\\
  
  If $\skel = \skneqtsub{\skel'}{\sksquant \tvar {\sksub}}$, then there exists
  $\typ'$ and $\typ$ such that $\jsub {\sksub}{\typ'} \vleq {\typ}$, $\jneq
  {\skel}{\lamb \expvar \term}{\envvar}{\uquant \tvar \typ}$, and $\jneq
  {\skel'}{\lamb \expvar \term}{\envvar}{\uquant \tvar {\typ'}}$. We have
  $\sz{\skel'} <\sz{\skel}$, so by induction there exists $\skel''$ such that
  $\tr{\skel'}=\skelquant \tvar {\skel''}$ and $\jneq {\tr{\skel'}}{\lamb
    \expvar \term}{\envvar}{\uquant \tvar {\typ'}}$. By rule \TRforall, we have
  $\jneq {\skel''}{\lamb \expvar \term}{\envvar}{\typ'}$. Therefore we have
  $\jneq {\skneqtsub {\skel''}{\sksub}}{\lamb \expvar \term}{\envvar}{\typ}$.

  By Lemma \ref{l:sz-decr}, we have $\sz{\tr{\skel'}} \leq \sz{\skel'}$, i.e.
  $\sz{\skel''}+1 \leq \sz{\skel'}$. By Lemma \ref{l:sz-prop}, we have
  $\sz{\skneqtsub {\skel''}{\sksub}} \leq \sz{\skneqtsub{\skel'}{\sksub}}$,
  hence we have $\sz{\skneqtsub{\skel''}{\sksub}} \leq
  \sz{\skneqtsub{\skel'}{\sksub}} < \sz\skel$. Consequently, by applying the
  induction hypothesis to $\skneqtsub{\skel''}{\sksub}$, we have $\jneq
  {\tr{\skneqtsub {\skel''}{\sksub}}}{\lamb \expvar \term}{\envvar}{\typ}$. By
  rule \TRforall we have $\jneq {\skelquant \tvar
    {\tr{\skneqtsub{\skel''}{\sksub}}}}{\lamb \expvar \term}{\envvar}{\uquant
    \tvar \typ}$. By definition of $\trsymb$ we have $\tr\skel=\skelquant \tvar
  {\tr{\skneqtsub{\skel''}{\sksub}}}$, hence the result holds.\\
  
  If $\skel = \skneqtsub {\skel'}{\sksquantcomm{\typ}}$, then we have $\jneq
  {\skel}{\lamb \expvar \term}{\envvar}{\uquant {\tvar_2}{\uquant
      {\tvar_1}\typ}}$ and $\jneq {\skel'}{\lamb \expvar \term}{\envvar}{\uquant
    {\tvar_1}{\uquant {\tvar_2}\typ}}$. By induction we have $\tr
  {\skel'}=\skelquant {\tvar_1}{\skel_1}$ with $\jneq {\tr{\skel'}}{\lamb
    \expvar \term}{\envvar}{\uquant {\tvar_1}{\uquant {\tvar_2}\typ}}$.
  Therefore we have $\jneq {\skel_1}{\lamb \expvar \term}{\envvar}{\uquant
    {\tvar_2}\typ}$, so by induction we have $\tr{\skel_1}=\skelquant
  {\tvar_2}{\skel_2}$ with $\jneq {\tr{\skel_1}}{\lamb \expvar
    \term}{\envvar}{\uquant {\tvar_2}\typ}$. Consequently we have $\jneq
  {\skel_2}{\lamb \expvar \term}{\envvar}{\typ}$, so by rule \TRforall, we have
  $\jneq {\skelquant {\tvar_2}{\skelquant {\tvar_1}{\skel_2}}}{\lamb \expvar
    \term}{\envvar}{\uquant {\tvar_2}{\uquant {\tvar_1} \typ}}$, and we
  have $\tr \skel = \skelquant {\tvar_2}{\skelquant {\tvar_1}{\skel_2}}$, as
  required.\\

  If $\skel = \skneqtsub {\skel'}{\sksdummyin{\typ}}$, then we have $\jneq
  {\skel}{\lamb \expvar \term}{\envvar}{\uquant \tvar \typ}$ with $\jneq
  {\skel'}{\lamb \expvar \term}{\envvar} \typ$ and $\tvar \notin \ftv
  {\skel'}$. Using rule \TRforall, we obtain $\jneq {\uquant \tvar
    {\skel'}}{\lamb \expvar \term}{\envvar}{\uquant \tvar \typ}$, i.e., $\jneq
  {\tr\skel}{\lamb \expvar \term}{\envvar}{\uquant \tvar \typ}$, as required.\\

  If $\skel = \skneqtsub {\skel'}{\sksdummyelim{\typ}}$, then we have $\jneq
  {\skel}{\lamb \expvar \term}{\envvar} \typ$ with $\jneq {\skel'}{\lamb \expvar
    \term}{\envvar}{\uquant \tvar \typ}$ and $\tvar \notin \ftv \envvar$. By
  induction, there exists $\skel_1$ such that $\tr{\skel'} = \skelquant \tvar
  {\skel_1}$ and $\jneq {\uquant \tvar {\skel_1}}{\lamb \expvar
    \term}{\envvar}{\uquant \tvar \typ}$. Consequently we have $\jneq
  {\skel_1}{\lamb \expvar \term}{\envvar}{ \typ}$, and because $\sz{\skel_1} <
  \sz {\skel}$, we obtain $\jneq {\tr{\skel_1}}{\lamb \expvar
    \term}{\envvar}{\typ}$ by induction (and the shape of $\tr{\skel_1}$ matches
  the one of $\typ$). Because $\tr \skel = \tr{\skel_1}$, we have the required
  result.\\

  If $\skel = \skneqtsub {\skel'}{\sksfun{\sksub_1}{\sksub_2}}$, then there
  exists $\typ'_1$, $\typ'_2$, $\typ_1$, and $\typ_2$ such that $\jsub
  {\sksub_1}{\typ_1} \leqeq {\typ'_1}$, $\jsub{\sksub_2}{\typ'_2} \leqeq {\typ_2}$, and $\jneq
  \skel {\lamb \expvar \term} \envvar {\fun{\typ_1}{\typ_2}}$. By the
  type system, we have $\jneq {\skel'}{\lamb \expvar \term} \envvar
  {\fun{\typ'_1}{\typ'_2}}$. We have
  $\sz{\skel'}<\sz{\skel}$, so by induction there exists $\skel''$ such that
  $\tr{\skel'}=\lamb \expvar {\skel''}$ and $\jneq {\lamb \expvar {\skel''}}{\lamb
    \expvar \term} \envvar {\fun{\typ'_1}{\typ'_2}}$. By the type
  system we have $\jneq {\skel''}{\term}{\eextend \envvar {\eassign \expvar 
    {\typ'_1}}}{\typ'_2}$, so by rules \TRsolvedsub and \TRenvsub we
  have $\jneq {\skneqenvsub {(\skneqtsub \skel {\sksub_2})}{\eassign \expvar {\sksub_1}
  }}{\term}{\eextend \envvar {\eassign \expvar {\typ_1}}}{\typ_2}$. Consequently we have $\jneq
  {\lamb \expvar {(\skneqenvsub {(\skneqtsub \skel {\sksub_2})}{\eassign \expvar
        {\typ_1}})}}{\lamb
    \expvar \term}{\envvar}{\fun{\typ_1}{\typ_2}}$, and $\tr\skel = \lamb \expvar
  {{\skneqenvsub {(\skneqtsub \skel {\sksub_2})}{\eassign \expvar {\typ_1}}}}$, hence the result holds.\\

  If $\skel=\skneqtsub \skel {\skssubv \subv \settv \sksub}$, then we have
  $\jneq {\skel}{\lamb \expvar \term}{\envvar}{\styp \subv \settv {\typ_2}}$ and
  $\jneq {\skel'}{\lamb \expvar \term}{\envvar}{\styp \subv \settv {\typ_1}}$
  with $\jsub {\sksub}{\typ_1} \leqeq {\typ_2}$. By induction we have $\tr
  {\skel} = \skels \subv \settv {\skel''}$ and $\jneq {\tr{\skel'}}{\lamb
    \expvar \term}{\envvar}{\styp \subv \settv {\typ_1}}$. Therefore we have
  $\jneq {\skel''}{\lamb \expvar \term}{\envvar}{\typ_1}$, so by rule
  \TRsolvedsub we have $\jneq {\skneqtsub{{\skel''}}\sksub}{\lamb \expvar
    \term}{\envvar}{\typ_2}$, and by rule \TRsintro $\jneq {\skels \subv \settv
    {(\skneqtsub {(\skel'')}\sksub)}}{\lamb \expvar \term}{\envvar}{\styp \subv
    \settv {\typ_2}}$.  Since we have $\tr \skel = \skels \subv \settv
  {(\skneqtsub {(\skel'')}\sksub)}$, the result holds.
  
\end{proof}

\subsection{Subject reduction}

\begin{lemma}
  \label{l:redex}
  If $\jneq {\skel_1}{\term}{\eextend \envvar {\eassign \expvar
      {\typ_1}}}{\typ_2}$ and $\jneq {\skel_2}{\vvar}{\envvar}{\typ_1}$, then
  there exists $\skel'$ such that $\jneq {\skel'}{\subst {\assign \expvar
      {\vvar}}{\term}}{\envvar}{\typ_2}$.
\end{lemma}

\begin{proof}
  We proceed by induction on $\skel_1$.

  Suppose $\skel_1 = \skelvar \expvar {\eextend \envvar {\eassign \expvar
      {\typ_1}}}$; we have $\jneq
  {\skel_1}{\expvar}{\envvar \sepenv \eassign \expvar {\typ_1}}{\apply \envvar
    \expvar}$ with $\typ_1 = \typ_2 = \apply \envvar \expvar$. We have $\jneq
  {\skel_2}{\subst {\assign \expvar {\vvar}}{\expvar}}{\envvar}{\typ_1}$, hence
  the result holds.\\
  
  Suppose $\skel_1 = \skelvar y {\eextend \envvar {\eassign \expvar
      {\typ_1}}}$ with $y \neq \expvar$; we have $\jneq
  {\skel_1}{y}{\envvar \sepenv \eassign \expvar {\typ_1}}{\apply \envvar y}$. We
  have $\subst {\assign \expvar \vvar} y = y$ and $\jneq
  {\skelvar y \envvar}{y}{\envvar}{\apply \envvar y}$, therefore we have the required result.\\

  Suppose $\skel_1 = \skelabs y {\skel'_1}$; we have $\jneq
  {\skel'_1}{\term'}{\eextend {\eextend {\envvar}{\eassign \expvar
        {\typ_1}}}{\eassign y {\typ_2^1}}}{\typ_2^2}$ with $\term = \lamb y
  {\term'}$ and $\typ_2 = \fun {\typ_2^1}{\typ_2^2}$. By induction, we have
  $\jneq {\skel'}{\subst {\assign \expvar {\vvar}}{\term'}}{\envvar
      \sepenv \eassign y {\typ_2^1}}{\typ_2^2}$.  By rule \TRabs, we have
  $\jneq {\skelabs y {\skel''}}{\lamb y {(\subst {\assign \expvar
        {\vvar}}{\term'})}}{\envvar}{\fun {\typ_2^1}{\typ_2^2}}$, as required.\\

  Suppose $\skel_1 = \skel_1^1 \skelapp \skel_1^2$; we have $\term= \term_1 \app
  \term_2$ with $\jneq {\skel_1^1}{\term_1}{\eextend {\envvar}{\eassign \expvar
      {\typ_1}}}{\fun{\typ_3}{\typ_2}}$, and $\jneq
  {\skel_1^2}{\term_2}{\eextend {\envvar}{\eassign \expvar
      {\typ_1}}}{\typ_3}$. By induction we have $\jneq {\skel'_1}{\subst
    {\assign \expvar {\vvar}}{\term_1}}{\envvar}{\fun {\typ_3}{\typ_2}}$ and
  $\jneq {\skel'_2}{\subst {\assign \expvar
      {\vvar}}{\term_2}}{\envvar}{\typ_3}$, so by rule \TRapp we have $\jneq
  {\skel'_1 \skelapp \skel'_2}{\subst {\assign \expvar {\vvar}}{\term_1} \app
    \subst {\assign \expvar {\vvar}}{\term_2}}{\envvar}{\typ_2}$, i.e.  $\jneq
  {\skel'_1 \skelapp \skel'_2}{\subst {\assign \expvar {\vvar}}{(\term_1 \app
      \term_2)}}{\envvar}{\typ_2}$, as required.\\

  Suppose $\skel_1 = \skelquant \tvar {\skel'_1}$; we have $\jneq
  {\skel'_1}{\term}{\eextend {\envvar}{\eassign \expvar {\typ_1}}}{\typ'_2}$
  with $\typ_2 = \uquant \tvar {\typ'_2}$. By induction we have $\jneq
  {\skel''_1}{\subst {\assign \expvar {\vvar}}{\term}}{\envvar}{\typ'_2}$, so by
  rule \TRforall we have $\jneq {\skelquant \tvar {\skel''_1}}{\subst {\assign
      \expvar {\vvar}}{\term} }{\envvar}{\uquant \tvar {\typ'_2}}$.\\

  Suppose $\skel_1 = \skneqtsub {{\skel'}_1}{\sksub}$; we have $\jneq
  {\skel'_1}{\term}{\eextend {\envvar}{\eassign \expvar {\typ_1}}}{\typ'_2}$
  with $\jsub {\sksub}{\typ'_2} \vleq {\typ_2}$. By induction there exists
  $\skel'$ such that $\jneq {\skel'}{\subst {\assign \expvar
      {\vvar}}{\term}}{\envvar}{\typ'_2}$, so by rule \TRsolvedsub we have
  $\jneq {\skneqtsub {{\skel'}}{\sksub}}{\subst {\assign \expvar
      {\vvar}}{\term}}{\envvar}{\typ_2}$.\\
  
  Suppose $\skel_1 = \skneqenvsub {{\skel'}_1} {\eassign y \sksub}$ with $y \neq
  \expvar$; we have $\jneq {\skel'_1}{\term}{\eextend {\eextend
      {\envvar'}{\eassign \expvar {\typ_1}}}{\eassign y {\typ'_3}}}{\typ_2}$
  with $\envvar=\eextend {\envvar'}{\eassign y {\typ_3}}$ and $\jsub \sksub
  {\typ_3} \leqeq {\typ'_3}$. By induction there exists $\skel'$ such that
  $\jneq {\skel'}{\subst{\assign \expvar \vvar}\term}{\eextend
    {\envvar'}{\eassign y {\typ'_3}}}{\typ_2}$. By rule \TRenvsub, we have
  $\jneq {\skneqenvsub {{\skel'}}{\eassign y {\sksub'}}}{\subst{\assign \expvar
      \vvar}\term}{\eextend {\envvar'}{\eassign y {\typ_3}}}{\typ_2}$,
  i.e. $\jneq {\skneqenvsub {{\skel'}}{\eassign y {\sksub'}}}{\subst{\assign
      \expvar
      \vvar}\term}{\envvar}{\typ_2}$, as required.\\

  Suppose $\skel_1 = \skneqenvsub {{\skel'}_1}{\eassign \expvar {\sksub}}$; we
  have $\jneq {\skel'_1}{\term}{\eextend {\envvar}{\eassign \expvar
      {\typ'_1}}}{\typ_2}$ with $\jsub \sksub {\typ_1} \leqeq {\typ'_1}$. By
  rule \TRsolvedsub, we have $\jneq {\skneqtsub
    {\skel_2}{\sksub}}{\vvar}{\envvar}{\typ'_1}$, so by induction there exists
  $\skel'$ such that $\jneq {\skel'}{\subst{\assign \expvar
      \vvar}\term}{\envvar}{\typ_2}$, hence the result holds.\\

  Suppose $\skel_1 = \skels \subv \settv {{\skel'}_1}$; we have $\jneq
  {\skel'_1}{\term}{\eextend {\envvar}{\eassign \expvar {\typ_1}}}{\typ'_2}$
  with $\typ_2 = \styp \subv \settv {\typ'_2}$ and $\ftv{\envvar} \cup
  \ftv{\typ_1} \subseteq \settv$. By induction there exists $\skel'$ such that
  $\jneq {\skel'}{\subst {\assign \expvar
      {\vvar}}{\term}}{\envvar}{\typ'_2}$. We have $\ftv{\envvar} \subseteq
  (\ftv{\envvar} \cup \ftv{\typ_1}) \subseteq \settv$, so by rule \TRsintro, we
  have $\jneq {\skels \subv \settv {{\skel'}}}{\subst {\assign \expvar
      {\vvar}}{\term}}{\envvar}{\typ_2}$, as required.

\end{proof}

\begin{lemma}
  \label{l:subj-reduction-neq}
  If $\jneq \skel \term \envvar \typ$ and $\term \cbv \term'$ then there
  exists $\skel'$ such that $\jneq {\skel'}{\term'} \envvar \typ$. 
\end{lemma}

\begin{proof}
  The proof is by induction on $\skel$. In the application case, we have $\term
  = \term_1 \app \term_2$, $\skel = \skel_1 \app \skel_2$, $\jneq
  {\skel_1}{\term_1}{\envvar}{\fun {\typ'}{\typ}}$ and $\jneq
  {\skel_2}{\term_2}{\envvar}{\typ'}$. We proceed by induction on
  $\term \cbv \term'$.

  If the $\beta$-rule is applied, then $\term_1 = \lamb \expvar {\term_3}$,
  $\term_2$ is a value $\vvar$, and $\term'=\subst {\assign \expvar
    \vvar}{\term_3}$. By Lemma \ref{l:tr-abs}, there exists $\skel'_1$ such that
  $\jneq {\skelabs \expvar {\skel'_1}}{\lamb \expvar {\term_3}}{\envvar}{\fun
    {\typ'}{\typ}}$.  By rule \TRabs, we have $\jneq
  {\skel'_1}{\term_3}{\eextend {\envvar}{\eassign \expvar {\typ'}}}{\typ}$. We
  have then the required result by Lemma \ref{l:redex}.

  If the first congruence rule is applied, we have $\term_1 \cbv \term'_1$ and $\term'
  = \term'_1 \app \term_2$. By the induction hypothesis on the reduction, there
  exists $\skel'_1$ such that $\jneq {\skel'_1}{\term'_1}{\envvar}{\fun{\typ'} \typ}$. By
  rule \TRapp we have $\jneq {\skel'_1 \skelapp
    \skel_2}{\term'}{\envvar} \typ$, hence the result holds.

  If the second congruence rule is applied, we have $\term_2 \cbv \term'_2$,
  $\term_1$ is a value $\vvar$ and $\term' = \vvar \app \term'_2$. By the
  induction hypothesis on the reduction, there exists $\skel'_2$ such that
  $\jneq {\skel'_2}{\term'_2}{\envvar}{\typ'}$. By rule \TRapp we have $\jneq
  {\skel_1 \skelapp \skel'_2}{\term'}{\envvar} \typ$, hence the result holds.\\

  The other cases (type constructor introductions, subtypings) are
  straightforward by induction.
  
\end{proof}

\begin{theorem}
  If $\jtyp \skel \term \envvar \typ \cons$ where $\cons$ is solved, and $\term
  \cbv \term'$, then there exists $\skel'$ and a solved $\cons'$ such that
  $\jtyp {\skel'}{\term'} \envvar \typ {\cons'}$.

\end{theorem}

\begin{proof}
  By Lemma \ref{l:corr-type-systems} we have $\jneq {\skel'} \term \envvar
  \typ$.  By Lemma \ref{l:subj-reduction-neq}, we have $\jneq {\skel'}{\term'}
  \envvar \typ$, so by Lemma \ref{l:corr-type-systems}, we have $\jtyp
  {\skel''}{\term'}{\envvar}{\typ}{\cons'}$, hence the result holds.

\end{proof}


\section{Initial skeletons}

We now prove that we can generate \systemfs skeletons from an initial
skeleton. 

\begin{lemma}
  \label{l:generate-fs-from-init}
  Let $\skel$ such that $\jinit \xtoa \term \skel$, and let $\skel^1$ such that
  $\jtyp {\skel^1} \term \envvar \typ \cons$ and $\support \envvar =
  \support \xtoa$. There exists $\svar$ such that $\support \svar = \allvar
  \skel$ and $\jtyp{ \subst{\svar} \skel} \term \envvar \typ {(\cons \consinter
    \cons')}$ where $\cons'$ is reflexive.

\end{lemma}

\begin{proof}
  By induction on $\skel^1$.

  If $\skel^1 = \skelvar \expvar \envvar$, then we have $\jtyp {\skelvar \expvar
    \envvar} \expvar \envvar {\apply \envvar \expvar} \consomega$ and $\skel =
  \skels \subv {\{ \tvar \}}{\skelvar \expvar \tvar}$ for some $\subv$ and
  $\tvar$. Let $\svar$ be the substitution which substitutes $\apply \envvar
  {\expvar_i}$ for $\apply \xtoa {\expvar_i}$ for all $i$, and such that $\subst
  \svar \subv = \idi$. We have $\subst \svar
  \skel = \skelvar \expvar \envvar$, hence the result holds.\\

  \begin{sloppypar}
    If $\skel^1 = \skelabs \expvar {{\skel^1}'}$, then we have $\jtyp {\skelabs
      \expvar {{\skel^1}'}}{\lamb \expvar \term} \envvar {\fun{\typ_1}{\typ_2}}
    \cons$ with $\jtyp {{\skel^1}'} \term {\eextend \envvar {\eassign \expvar
        {\typ_1}}}{\typ_2} \cons$. We also have $\skel = \skels {\subv}{\ftv
      {\tenv {\skelabs \expvar {\skel'}}}}{\skelabs \expvar {\skel'}}$ with
    $\jinit {\xtoa \sepenv \eassign \expvar \tvar} \term {\skel'}$ and $\subv
    \notin \allvar {\skel'}$. By induction there exists $\svar$ such that
    $\jtyp{\subst \svar {\skel'}} \term {\envvar \sepenv \eassign \expvar
      {\typ_1}}{\typ_2}{(\cons \consinter \cons')}$ with $\cons'$ reflexive. Let
    $\svar' = (\svar \sepsubs \assign {\subv} \idi \sepsubs \idsubs)$. We have
    $\jtyp{\subst{\svar'}{\skel'}} \term {\envvar \sepenv \eassign \expvar
      {\typ_1}}{\typ_2}{(\cons \consinter \cons')}$, so by rule \TRabs, we have
    $\jtyp{\skelabs \expvar {\subst{\svar'}{\skel'}}} \term
    {\envvar}{\fun{\typ_1}{\typ_2}} {(\cons \consinter \cons')}$. We have
    $\subst {\svar'}{\skel}=
    \skelabs \expvar {\subst{\svar'}{\skel'}}$, hence the result holds.\\
  \end{sloppypar}

  If $\skel^1 = {\skel^1}_1 \skelapp {\skel^2}_2$, then we have $\jtyp
  {{\skel^1}_1 \skelapp {\skel^1}_2}{\term_1 \app \term_2}{\envvar}{\typ_2}
  \cons$ with $\jtyp {{\skel^1}_1}{\term_1}{\envvar}{\fun{\typ_1}{\typ_2}}
  {\cons_1}$ and $\jtyp {{\skel^1}_2}{\term_2}{\envvar}{\typ_1}{\cons_2}$. We
  also have $\skel = \skels {\subv}{\ftv {\tenv {\skel'}}}{\skel'}$ with $\skel'
  = \sktsub{\skel_1}{\fun{\rtype{\skel_2}}{\tvar}} \skelapp \skel_2$, $\skel_1$
  initial skeleton for $\term_1$ and $\skel_2$ initial skeleton for $\term_2$.
  By induction there exists $\svar_1$, $\svar_2$ such that
  $\jtyp{\subst{\svar_1} {\skel_1}}{\term_1}{\envvar}{\fun{\typ_1}{\typ_2}}
  {(\cons_1 \consinter \cons'_1)}$ and $\jtyp{\subst{\svar_2}
    {\skel_2}}{\term_2}{\envvar}{\typ_1} {(\cons_2 \consinter \cons'_2)}$ with
  $\cons'_1$, $\cons'_2$ reflexive. Let $\svar = (\svar_1 \sepsubs \svar_2
  \sepsubs \assign{\subv} \idi \sepsubs \assign \tvar {\typ_2} \sepsubs
  \idsubs)$.  We have $\jtyp{\subst \svar
    {\skel_1}}{\term_1}{\envvar}{\fun{\typ_1}{\typ_2}}{(\cons_1 \consinter
    \cons'_1)}$, $\jtyp{\subst \svar
    {\skel_2}}{\term_2}{\envvar}{\typ_1}{(\cons_2 \consinter \cons'_2)}$, and
  $\subst \svar
  {(\fun{\rtype{\skel_1}}{\tvar})}=\fun{\typ_1}{\typ_2}$. Consequently we have
  $\jtyp{\sktsub{\subst \svar {\skel_1}}{\subst \svar
      {(\fun{\rtype{\skel_1}}{\tvar})}} \skelapp \subst \svar {\skel_2}}{\term_1
    \app \term_2}{\envvar}{\typ_2}{\cons''}$ with $\cons'' = \cons_1 \consinter
  \cons_2 \consinter \cons'_1 \consinter \cons'_2 \consinter
  (\fun{\typ_1}{\typ_2} \consleq \fun{\typ_1}{\typ_2})$, and $\subst \svar \skel
  = \sktsub{\subst \svar {\skel_1}}{\subst \svar
    {(\fun{\rtype{\skel_1}}{\tv{j''+1}})}} \skelapp \subst \svar {\skel_2}$,
  hence the result holds.\\

  If $\skel^1 = \skelquant \tvar {{\skel^1}'}$, then we have $\jtyp {\skelquant
    \tvar {{\skel^1}'}} \term {\envvar}{\uquant \tvar \typ}{\consquant \tvar
    \cons}$ with $\jtyp {{\skel^1}'} \term \envvar \typ \cons$ and $\tvar \notin
  \ftv \envvar$. We have $\skel = \skels \subv \settv {\skel'}$ for some
  $\subv$, $\settv$, and $\skel'$. By induction there exists $\svar$ such that
  $\jtyp{ \subst{\svar} \skel} \term \envvar {\typ}{(\cons \consinter \cons')}$
  with $\cons'$ reflexive. Let $\Ivar = \subst \svar {\subv}$. Let $\svar'$ be
  the substitution equal to $\svar$ except on $\subv$, where $\subst
  {\svar'}{\subv} = \iquant \tvar \Ivar$. We have $\jtyp{\subst {\svar'} \skel}
  \term \envvar {\uquant \tvar \typ}{(\cons \consinter \consquant \tvar
    {\cons'})}$, hence the result holds.\\

  If $\skel^1 = \sktsub {({\skel^1}')}{\typ_2}$, then we have $\jtyp {\sktsub
    {({\skel^1}')} {\typ_2}} \term {\envvar}{\typ_2}{(\cons^1 \consinter \typ_1
    \consleq \typ_2)}$ with $\jtyp {{\skel^1}'} \term \envvar
  {\typ_1}{\cons^1}$. We have $\skel = \skels \subv \settv {\skel'}$ for some
  $\subv$, $\settv$, and $\skel'$. By induction there exists $\svar$ such that
  $\jtyp{\subst{\svar} \skel} \term \envvar {\typ_1}{(\cons^1 \consinter
    \cons'^1)}$ with $\cons'^1$ reflexive. Let $\Ivar = \subst \svar
  {\subv}$. Let $\svar'$ be the substitution equal to $\svar$ except on $\subv$,
  where $\subst {\svar'}{\subv} = \isub \Ivar {\typ_2}$. We have
  $\jtyp{\subst{\svar'} \skel} \term \envvar {\typ_2}{(\cons^1 \consinter
    (\typ_1 \consleq \typ_2)
    \consinter \cons'^1)}$, hence the result holds.\\

  If $\skel^1 = \skels \subv \settv {{\skel^1}'}$, then we have $\jtyp {\skels
    \subv \settv {{\skel^1}'}} \term {\envvar}{\styp \subv \settv \typ}
  {\conss \subv \settv \typ {\cons^1}}$ with $\jtyp {{\skel^1}'}
  \term \envvar {\typ}{\cons^1}$. We have $\skel = \skels
  {\subv'}{\settv'}{\skel'}$ for some $\skel'$, $\subv'$, and $\settv'$. By
  induction there exists $\svar$ such that $\jtyp{\subst{\svar} \skel} \term
  \envvar {\typ_1}{(\cons^1 \consinter \cons'^1)}$ with $\cons'^1$
  reflexive. Let $\Ivar = \subst \svar {\subv'}$. Let $\svar'$ be the
  substitution equal to $\svar$ except on $\subv'$, where $\subst
  {\svar'}{\subv'} = \iapp \subv {\settv \setminus \ftv{\subst
      \svar {\settv'}}} \Ivar$. We have $\jtyp{ \subst{\svar'} \skel} \term \envvar
  {\styp \subv \settv \typ}{(\conss \subv \settv \typ {\cons^1} \consinter 
    \conss \subv \settv \typ {\cons'^1})}$, hence the result holds.

\end{proof}

\begin{theorem}
  \label{t:generate-fs-from-init}
  Let $\skel$ such that $\jinitsk \term \skel$, and let $\skel^1$ be a relevant
  skeleton such that $\jtyp {\skel^1} \term \envvar \typ \cons$. There
  exists $\svar$ such that and $\jtyp{ \subst{\svar} \skel} \term \envvar \typ
  {(\cons \consinter \cons')}$ with $\cons'$ reflexive.

\end{theorem}

\begin{proof}
  If $\jinitsk \term \skel$, then there exists $\xtoa$ such that $\jinit \xtoa
  \term \skel$ and $\support \xtoa = \fv \term$. The skeleton $\skel^1$ is
  relevant so we have $\support \envvar = \fv \term$. Therefore we have
  $\support \xtoa = \support \envvar$, and we have the required result by Lemma
  \ref{l:generate-fs-from-init}.
\end{proof}


\end{document}